\documentclass{article}
\usepackage{amsmath,amsthm,amssymb}
\makeatletter
\renewcommand*{\@fnsymbol}[1]{\ensuremath{\ifcase#1\or 1\or 2\or
 3\or 4\or 5\or 6\or 7 \or 8 \else\@ctrerr\fi}}
\makeatother
\usepackage{cite}

\DeclareMathOperator*{\esssup}{\mathrm{ess\,sup}}

\allowdisplaybreaks[3]
\newcommand{\eps}{\varepsilon}

\newenvironment{myproof}[1]{\begin{proof}[Proof of #1]}{\end{proof}}
\numberwithin{equation}{section}
\newtheorem{theorem}{Theorem}[section]
\newtheorem{corollary}[theorem]{Corollary}
\newtheorem{lemma}[theorem]{Lemma}
\newtheorem{proposition}[theorem]{Proposition}
\theoremstyle{definition}
\newtheorem{definition}[theorem]{Definition}

\theoremstyle{remark}
\newtheorem{remark}[theorem]{Remark}

\title{The statistical dynamics of a spatial logistic model
       and the related kinetic equation}

\author{Dmitri Finkelshtein\thanks{Department of Mathematics,
Swansea University, Singleton Park, Swansea SA2 8PP, U.K. ({\tt d.l.finkelshtein@swansea.ac.uk}).} \and Yuri Kondratiev\thanks{Fakult\"{a}t
f\"{u}r Mathematik, Universit\"{a}t Bielefeld, Postfach 110 131, 33501 Bielefeld,
Germany ({\tt kondrat@math.uni-bielefeld.de}).} \and Yuri Kozitsky\thanks{Instytut Matematyki, Uniwersytet Marii Curie-Sk{\l}odwskiej, 30-031 Lublin, Poland ({\tt jkozi@hektor.umcs.lublin.pl}).} \and Oleksandr
Kutoviy\thanks{Department of Mathematics, Massachusetts Institute of Technology,
77 Massachusetts Avenue E18-420, Cambridge, MA 02139, USA ({\tt
kutovyi@mit.edu}); Fakult\"{a}t f\"{u}r Mathematik, Universit\"{a}t
Bielefeld, Postfach 110 131, 33501 Bielefeld, Germany ({\tt
kutoviy@math.uni-bielefeld.de}).}}

\begin{document}

\maketitle

\begin{abstract}

There is studied an infinite system of
point entities in $\mathbb{R}^d$ which reproduce themselves and die, also due to competition. The system's states  are
probability measures on the space of configurations of entities. Their evolution
is described by means of a BBGKY-type equation for the corresponding correlation (moment)
functions. It is proved that: (a)  these functions evolve on a bounded time interval and  remain
sub-Poissonian due to the competition; (b) in the Vlasov
scaling limit they converge to the correlation functions of the time-dependent Poisson point
field  the density of which solves the
kinetic equation obtained  in the scaling limit from the equation for the correlation functions.
A number of properties of the solutions of the kinetic equation  are also established.



\end{abstract}

\section{The Setup } \label{Sec1}
\subsection{Introduction}
In life sciences, one often  deals with large systems of interacting
entities distributed over a continuous habitat and evolving in time,
cf. \cite{Bellomo,Neuhauser}. Their collective behavior is
observed at a macro-scale, and thus the mathematical theories  traditionally describe this behavior  by means of
phenomenologically deduced nonlinear equations involving macroscopic characteristics like
density, mobility, etc, see, e.g., \cite{zhao}. However,
this kind of macroscopic phenomenology may often be insufficient as
it does not take into account individual behavior of the constituting
entities.
 Thus, mathematical
models and methods are needed for drawing
population-level conclusions from individual-level descriptions.
The present paper is aimed at contributing to the development of the
theory of this kind.
We continue studying the  model
introduced and discussed in \cite{BP1,BP3,DL,FKK-MFAT,Dima,FM,Mu}.
This model describes a population of entities (e.g., perennial
plants) distributed over $\mathbb{R}^d$, which reproduce themselves
and die, also due to competition. As was suggested already in \cite{BP1}, see also page 1311 in \cite{Neuhauser}, the
 mathematical context proper for studying such objects is provided by the
theory of random point fields on $\mathbb{R}^d$. In this setting,
populations are modeled as point configurations constituting the set
\begin{equation} \label{C1}
 \Gamma :=\{\gamma\subset\mathbb R^d :
 |\gamma\cap K|<\infty\text{ for any compact $K\subset\mathbb R^d$
 }\},
\end{equation}
where $|A|$ stands for the number of elements in $A$. This will be
the phase space of our model. Along with finite configurations it
contains also infinite ones, which allows for studying `bulk'
properties ignoring boundary and size effects.

In the Hamiltonian mechanics, the motion of $N$ physical particles
in $\mathbb{R}^d$ is described by a system of $2dN$ differential
equations. For $N \gg 1$ (Avogadro's number is $\simeq6\times
10^{23}$), the point-wise description gets  meaningless since no
observation could indicate at which point of the phase space the
system actually is. Moreover, the description in terms of individual
trajectories would be `too detailed' to yield understanding the
collective behavior of the system. It was realized already in the
time of A. Einstein and M. Smoluchowski that the statistical
approach in the theory of such systems can link to each other their
micro- and macroscopic descriptions. In this approach, one deals
with the probabilities with which points of the phase space lie in
its subsets. The corresponding probability measures are then
considered as the states of the system. However, for interacting
particles, the direct study of the evolution of such states
encounters serious technical difficulties. In \cite{Bogol},
N. N. Bogoliubov suggested to do this by means of the so-called
correlation (moment) functions. Their evolution is obtained from an
infinite system of equations \cite{Dob}, called now BBGKY hierarchy
or chain, that links to each other correlation functions of
different order. Starting from the late 1990'th, a similar
statistical approach is being implemented in the dynamics of states
on $\Gamma$, see \cite{FKK} and the references therein.
Gradually, it has become clear also for theoretical
biologists\cite{BCFKKO} that the theory developed in this framework
can provide effective methods for studying individual-based models
of large systems of living entities.

In this work, the evolution of states $\mu_0 \mapsto \mu_t$ on $\Gamma$  is described by
the Fokker-Planck equation\footnote{For further details on the dynamics of states on
$\Gamma$ see \cite{Berns,FKKo,FKK-MFAT,FKK,DimaN,Dima,Dima2}.}
\begin{equation}
  \label{Fokk}
\frac{d}{dt} \mu_t = L^* \mu_t, \qquad  \mu_t|_{t=0} = \mu_0, \qquad t>0,
\end{equation}
in which `operator' $L^*$ specifies the model. One can also study
the evolution of {\it observables} -- real valued functions on
$\Gamma$,
 by means of the Kolmogorov equation
\begin{equation}
 \label{R2}
\frac{d}{d t} F_t = L F_t , \qquad F_t|_{t=0} = F_0, \qquad t>0.
\end{equation}
`Operators'
 $L$ and $L^*$ are related to each other in such a way that the solutions of (\ref{Fokk}) and (\ref{R2})
satisfy
\begin{equation*}
\int_{\Gamma} F_0 d \mu_t  = \int_{\Gamma} F_t d \mu_0, \qquad t>0.
\end{equation*}
For the model studied in this work, $L$  has
the form
\begin{eqnarray}\label{R20}
(LF)(\gamma) &=& \sum_{x\in \gamma}\left[ m + E^{-} (x, \gamma\setminus x) \right]\left[F(\gamma\setminus x)
- F(\gamma) \right]\\[.2cm]
&& \qquad + \int_{\mathbb{R}^d} E^{+} (y, \gamma) \left[ F(\gamma\cup y) - F(\gamma) \right]dy, \nonumber
\end{eqnarray}
where
\begin{equation}
 \label{Ra20}
E^{\pm}(x, \gamma) := \sum_{y\in \gamma} a^{\pm} (x- y)\geq 0.
\end{equation}
The first term in (\ref{R20}) describes the death of the particle
located at $x$, occurring independently with rate $m\geq 0$
(intrinsic mortality) and under the influence of the other particles
in $\gamma$ (competition) with rate $E^{-}(x, \gamma\setminus x)$.
Here and in the sequel,  $x\in \mathbb{R}^d$ is also treated as a
single-point configuration.  The second term in (\ref{R20})
describes the birth of a particle at $y\in \mathbb{R}^d$ given by
the whole configuration $\gamma$ with rate $E^{+}(y, \gamma)$.

\subsection{Correlation functions}

As mentioned above, we shall construct the dynamics of the model by
employing  {\it correlation functions}, which fully characterize the
corresponding states.  Given $n\in \mathbb{N}$ and a probability
measure $\mu$, the $n$-th order correlation function $k^{(n)}_\mu$
is related to $\mu$ by the following formula
\begin{eqnarray}
  \label{KFm}
& & \int_{\Gamma} \left(\sum_{\{ x_1,\dots , x_n\} \subset \gamma}
G^{(n)}(x_1, \dots, x_n) \right) \mu(d \gamma)\\[.2cm]
& & \qquad = \frac{1}{n!}\int_{(\mathbb{R}^d)^n} G^{(n)}(x_1, \dots,
x_n)  k^{(n)}_\mu(x_1, \dots, x_n) dx_1 \cdots dx_n, \nonumber
\end{eqnarray}
which has to hold for all appropriate functions $G^{(n)}:
(\mathbb{R}^d)^n \to \mathbb{R}$, symmetric with respect to the
interchanges of its variables. Thus, each
$k^{(n)}_\mu:(\mathbb{R}^d)^n \to \mathbb{R}$ is a symmetric
function with a number of specific properties, see, e.g., \cite{FKK,DimaN,Dima,Dima2}. If one puts also $k^{(0)}_\mu \equiv
1$, then the collection of $k^{(n)}_\mu$, $n\in \mathbb{N}_0$,
determines a map, $k_\mu: \Gamma_0 \to \mathbb{R}$, defined on  the
set of finite configurations
\begin{equation}
\label{K10}
\Gamma_{0} = \bigsqcup_{n\in \mathbb{N}_0} \Gamma^{(n)} ,
\end{equation}
which is the disjoint union of the sets of $n$-particle configurations:
\begin{equation}
  \label{K10n}
\Gamma^{(0)} = \{ \emptyset\}, \qquad \Gamma^{(n)} = \{\eta \in \Gamma: |\eta| = n \}, \ \ n\in \mathbb{N}.
\end{equation}
Each $\Gamma^{(n)}$, $n\in \mathbb{N}$, can be equipped with the
topology related to the Euclidean topology of $\mathbb{R}^d$. The
restriction of $k_\mu$ to a given $\Gamma^{(n)}$, extended to a
symmetric function on $(\mathbb{R}^d)^n$, is exactly the $n$-th
order correlation function as in (\ref{KFm}). In particular,
$k^{(1)}_\mu$ is the density of the particles in state $\mu$. The
correlation function $k_{\pi_\varrho}$ of the inhomogeneous Poisson
measure $\pi_\varrho$ is
\begin{equation}
  \label{R400}
k_{\pi_\varrho} (\eta) = \prod_{x\in \eta} \varrho (x),
 \end{equation}
or, equivalently, $k^{(n)}_{\pi_\varrho} (x_1 , \dots , x_n) = \varrho(x_1) \cdots \varrho(x_n)$,
where the density $\varrho$  is supposed to be locally integrable.
A measure $\mu$ on $\mathcal{B}(\Gamma)$ is said to be {\it sub-Poissonian}
if its correlation function is such that, for for some $C>0$ and all $n\in \mathbb{N}$,
\begin{equation}
  \label{int-d}
 k_\mu^{(n)} (x_ 1 , \dots , x_n) \leq C^n, \qquad {\rm for} \ {\rm Lebesgue - a.a.} \ \ (x_1, \dots, x_n)\in(\mathbb{R}^d)^n.
\end{equation}
In a way, a sub-Poissonian state is similar to the Poissonian state
in which the particles are independently placed in $\mathbb{R}^d$.
At the same time, the increase of $k^{(n)}_\mu$  as $n!$, see
(\ref{99}) below, corresponds to the appearance of {\it clusters} in
state $\mu$.

By an appropriate procedure,\cite{Dima2} the Cauchy problem in
(\ref{R2}) is transformed into the following one
\begin{equation}
  \label{R4}
\frac{d}{d t} k_t = L^\Delta k_t, \qquad k_t|_{t=0} = k_0,
\end{equation}
where `operator'
\begin{align}
\label{ltrian}
(L^\Delta k)(\eta)  = & -\left[\sum_{x\in \eta}(m+E^{-}(x,\eta\setminus x))\right]k(\eta)+
\sum_{x\in\eta}E^{+}(x,\eta\setminus x) k(\eta\setminus x) \\[.2cm]
 &+ \int_{\mathbb R^{d}}\sum_{y\in\eta}a^{+}(x-y)k((\eta\setminus y)\cup
x)dx
-\int_{\mathbb R^{d}}E^{-}(x,\eta)k(\eta\cup x)dx \nonumber
\end{align}
is calculated from that in (\ref{R20}), and $k_0$ is the correlation
function of  $\mu_0$. In contrast to those in (\ref{R20}) and
(\ref{Ra20}), the sums in (\ref{ltrian}) are finite -- the advantage
of passing to (\ref{R4}). In terms of the `components' $k^{(n)}_t$,
the equation in (\ref{R4}) is an infinite chain of linked linear
equations, analogous to the BBGKY hierarchy mentioned above. The
first equations in (\ref{ltrian}) are: $dk^{(0)}_t /dt = 0$, and
\begin{align}
  \label{ltr}
\frac{d}{dt}k^{(1)}_t(x) = &- m k^{(1)}_t(x) - \int_{\mathbb{R}^d}
a^{-}(x-y) k^{(2)}_t(x,y) dy \\&+ \int_{\mathbb{R}^d} a^{+}(x-y)
k^{(1)}_t(y) dy .\nonumber
\end{align}
The equation with $dk^{(2)}_t /dt$ contains $k^{(n)}_t$
with $n = 1, 2, 3$, etc. Theoretical biologists `solve' such chains
by decoupling; cf. \cite{Mu}. In
the simplest version, one sets
\begin{equation}
\label{ltr0}
k_t^{(2)}(x, y)\simeq
 k_t^{(1)}(x)k_t^{(1)}(y),
\end{equation}
which amounts to neglecting spatial pair correlations (so
called {\it mean field} approximation). Thereafter, (\ref{ltr}) turns  into
the following nonlinear (closed) equation
\begin{eqnarray}
  \label{ltr1}
\frac{d}{dt}k_t^{(1)}(x) & = & - m k^{(1)}(x) - k_t^{(1)}(x) \int_{\mathbb{R}^d} a^{-}(x-y)
k_t^{(1)}(y) dy\\[.2cm] &&+ \int_{\mathbb{R}^d} a^{+}(x-y) k_t^{(1)}(y) dy , \nonumber
\end{eqnarray}
which is the kinetic
equation for our model, see  Section \ref{KSec} below.
Note that the question of whether the evolving states remain
sub-Poissonian, and hence clusters do not appear, can be answered
only by studying the whole chain $k^{(n)}_\mu$, $n\in \mathbb{N}_0$,
cf. (\ref{int-d}). For the contact model -- a particular case of the model we study corresponding
to $a^{-} \equiv 0$ in (\ref{R20}) and (\ref{ltrian}), it is
known\cite{Dima} that
\begin{equation}
\label{99}
{\rm const}\cdot n! c^n_t \leq
 k^{(n)}_t (x_1, \dots, x_n) \leq {\rm const}\cdot n! C^n_t,
\end{equation}
where the left-hand inequality holds if $x_i$ belong to a ball
of small enough radius.

\subsection{Mesoscopic description}

Along with the microscopic theory based on (\ref{R2}) and
(\ref{R4}), we provide in this work the mesoscopic description of
the evolution obtained from (\ref{R4}) by means of the Vlasov
scaling, see \cite{FKK-MFAT,DimaN}, and also Section 6 in
\cite{Dob} and  \cite{P} where the general aspects
of the scaling of interacting particle systems are discussed. In the
`physical language', the Vlasov scaling can be outlined as follows.
One considers the system at the scale where the particle density is
large, and hence the interaction should be respectively small in
order that the total interaction energy take intermediate values. In
the scaling limit, the corpuscular structure disappears and the
system turns into a medium described solely by the density, cf.
(\ref{R400}), whereas the interactions are taken into account in a
`mean-field-like' way. In this limit, the ansatz in (\ref{ltr0})
becomes exact and thus the evolution of the density is obtained from
(\ref{ltr1}). An important issue here is to control this passage in
a mathematically rigorous way, which includes  the convergence of
the `rescaled' correlation functions.

In this work, like also in \cite{Berns,FKK-MFAT,DimaN}, the
scale is described by a single parameter, $\varepsilon\in (0,1]$,
tending to zero in the scaling limit and taking value $\varepsilon
=1$ for the initial system described by (\ref{R4}). In order to get
high densities for small $\varepsilon$, we assume that the
correlation function $k_{0,\varepsilon}$ for small $\varepsilon$
behaves like $k_{0,\varepsilon}(\eta) \sim \varepsilon^{-|\eta|}$,
$\eta\in \Gamma_0$,  and thus the rescaled correlation function
$r_{0, \varepsilon} (\eta) = \varepsilon^{|\eta|}k_{0, \varepsilon}
(\eta)$, or equivalently
 \begin{equation}
 \label{ren-1}
r^{(n)}_{0, \varepsilon} (x_1 , \dots , x_n) = \varepsilon^n
k_{0,\varepsilon}^{(n)} (x_1 , \dots , x_n), \qquad n\in \mathbb{N},
\end{equation}
converges as $\varepsilon \to 0$ to the correlation function of a
certain state. Namely, we assume that $r^{(n)}_{0, \varepsilon} \to
r^{(n)}_0$ in $L^\infty ((\mathbb{R}^d)^n)$ for each $n\in
\mathbb{N}$. Next,  we rescale the interaction in (\ref{ltrian}) by
multiplying $a^{-}$ by $\varepsilon$, which yields
$L^\Delta_\varepsilon$ from $L^\Delta$ given in (\ref{ltrian}). This
means that the evolution $k_{0,\varepsilon} \mapsto
k_{t,\varepsilon}$ of the `dense' system
 is now
governed by (\ref{R4}) with $L^\Delta_\varepsilon$ in the
right-hand side.  We expect that the evolving system remains
`dense', and thus introduce the rescaled correlation functions
\begin{equation}\label{ordersing}
r_{t,\varepsilon}^{(n)}(x_, \dots, x_n) = \varepsilon^{n}
k_{t,\eps}^{(n)}(x_, \dots, x_n)
, \qquad n\in \mathbb{N},
\end{equation}
that  solve the following
Cauchy problem
\begin{equation}
  \label{RenK}
\frac{d}{d t} r_{t,\varepsilon} = L^\Delta_{\varepsilon, {\rm ren}}
r_{t,\varepsilon}, \qquad r_{t,\varepsilon}|_{t=0} =
r_{0},
\end{equation}
 which one derives from (\ref{R4}) by means of
(\ref{ren-1}) and (\ref{ordersing}). `Operator'
\begin{equation}
  \label{E1}
 L^\Delta_{\varepsilon,{\rm ren}} = R_\eps
 L^\Delta_{\varepsilon} R_{\eps^{-1}}, \qquad \left( R_{\varepsilon^{\pm 1}}k\right) \left( \eta \right) =\varepsilon^{\pm \left\vert \eta
\right\vert }k\left( \eta \right)
\end{equation}
has the following structure
\begin{equation}
  \label{RenK1}
L^\Delta_{\varepsilon,{\rm ren}} = V + \varepsilon C,
\end{equation}
where $V$ and $C$ are given in (\ref{22A}) below. Along with
(\ref{RenK}), it is natural to consider the Cauchy problem
\begin{equation}
  \label{E5}
  \frac{d}{d t} r_t = V r_t, \qquad r_t|_{t=0} = r_0,
\end{equation}
where $r_0$ is the same  as in (\ref{RenK}), and to expect that the
solution of (\ref{RenK}) converges to that of (\ref{E5}) as
$\varepsilon \to 0$. This would give interpolation between the cases
of $\varepsilon =1$ and $\varepsilon =0$, i.e., between (\ref{R4})
and (\ref{E5}). The main peculiarity of (\ref{E5}) is that the
evolution $r_0 \mapsto r_t$ obtained therefrom `preserves chaos'.
That is, if $r_0$ is the correlation function of the Poisson measure
$\pi_{\varrho_0}$, see (\ref{R400}), then, for all $t>0$ for which
one can solve  (\ref{E5}), the product form of (\ref{R400}) is
preserved, i.e., the solution is the product of the values of the
density $\varrho_t$ which solves
 the {\it kinetic
equation}, cf. (\ref{ltr1}),
\begin{eqnarray}
\label{V-eqn-gen} \frac{d}{dt}\varrho_t (x) & = &  - m \varrho_t(x)
 - \varrho_t(x) \int_{\mathbb{R}^d} a^{-}(x-y) \varrho_t (y) dy\\[.2cm]
 &+& \int_{\mathbb{R}^d} a^{+}(x-y) \varrho_t (y) dy , \qquad \varrho_t|_{t=0} = \varrho_0, \nonumber
\end{eqnarray}
\subsection{The aims of the paper}
For the model specified in \eqref{R20} we aim at:
\begin{itemize}
\item  proving that due to the competition the Cauchy problems in (\ref{RenK})
and in (\ref{E5}) have sub-Poissonian solutions on the same time
interval $[0,T_*)$ (done in Theorems \ref{2tm} and \ref{op-tm});
\item  proving that the solution of (\ref{E5}) has the product form of (\ref{R400})
with $\varrho_t$ which solves (\ref{V-eqn-gen}) (done in Lemma \ref{Vlm});
\item proving that the solutions of (\ref{RenK}) converge to that of (\ref{E5}),
i.e., $r_{t, \varepsilon} \to r_t$ as $\varepsilon \to 0$ (done in Theorem \ref{op1-tm});
\item proving  solvability and studying the solutions of the kinetic
equation (\ref{V-eqn-gen}) (done in Theorems \ref{R1tm} --
\ref{K4tm}).
\end{itemize}
In realizing this program, we partly follow the scheme developed in
\cite{Berns,FKKo}, based on Ovcyannikov's method,\cite{trev}
by means  of which the solutions are constructed in  scales of
Banach spaces. However, here we consider essentially different model
where this method cannot be applied directly. Instead, we elaborated
an original technique based on the use of `sun-dual' semigroups
acting in scales of Banach spaces perturbed by operators treatable
by Ovcyannikov's method,  see subsections \ref{OvcS} and
\ref{op-lmS} below, and especially Remark \ref{Ovcrk}. Further
comparison of our results and those obtained for the same model in
\cite{FKK-MFAT,Dima}  are given in subsection
\ref{Conclsec}.

\section{The Mathematical Framework and the Model}\label{Sec2}

For more details on the mathematics used in this paper, see \cite{Berns,FKKo,FKK-MFAT,Dima,Dima2,Tobi}.

By $\mathcal{B}(\mathbb{R}^d)$ and $\mathcal{B}_{\rm
b}(\mathbb{R}^d)$ we denote the set of all Borel and all bounded
Borel subsets of $\mathbb{R}^d$, respectively. The configuration
space (\ref{C1}) is endowed with the vague topology -- the weakest
topology that makes all the maps
\[
\Gamma \ni \gamma \mapsto
\int_{\mathbb{R}^d} f(x) \gamma (dx)= \sum_{x\in \gamma} f(x) , \quad f\in C_0 (\mathbb{R}^d),
\]
continuous. Here $C_0 (\mathbb{R}^d)$ stands for the set of all
compactly supported continuous functions $f:\mathbb{R}^d \rightarrow
\mathbb{R}$. The vague topology is metrizable in the way that makes
$\Gamma$ a complete and separable metric (Polish) space. By
$\mathcal{B}(\Gamma)$ we denote the corresponding Borel
$\sigma$-algebra.

The set  $\Gamma_0\subset \Gamma$, see (\ref{K10}) and (\ref{K10n}),
is endowed with the topology  of the disjoint union of the sets
$\Gamma^{(n)}$, each of which is endowed with the topology related
to the Euclidean topology of $\mathbb{R}^d$. By
$\mathcal{B}(\Gamma_{0})$ we denote  the corresponding Borel
$\sigma$-algebra. The vague topology of $\Gamma$ induces on
$\Gamma_0$ another topology, different from that just mentioned.
However, see Lemma~1.1 and Proposition~1.3 in \cite{Obata},
the corresponding Borel $\sigma$-algebras coincide, and hence $
\mathcal{B}(\Gamma_0) = \{ A \in \mathcal{B}(\Gamma): A \subset
\Gamma_0\}$. Therefore, a probability measure $\mu$ on
$\mathcal{B}(\Gamma)$ such that $\mu(\Gamma_0)=1$ can be redefined
as a measure on $\mathcal{B}(\Gamma_0)$. Moreover, a function $G:
\Gamma_0 \subset \Gamma \to \mathbb{R}$ is
$\mathcal{B}(\Gamma)/\mathcal{B}(\mathbb{R})$-measurable if and only
if its restrictions to each $\Gamma^{(n)}$, more precisely, the
functions $G^{(n)}:(\mathbb{R}^d)^n \to \mathbb{R}$ such that
\begin{equation*}
G^{(0)} =G(\emptyset)\in \mathbb{ R}, \quad  G^{(n)} (x_1 , \dots , x_n) =
G(\gamma) \quad {\rm for} \ \ \gamma = \{x_1 , \dots , x_n\}, \quad n \in \mathbb{N},
\end{equation*}
are symmetric and Borel.
By the expression
\begin{equation*}
\int_{\Gamma_0} G(\eta) \lambda (d \eta) = G^{(0)} + \sum_{n=1}^\infty
\frac{1}{n!} \int_{(\mathbb{R}^d)^n} G^{(n)} (x_1 , \dots , x_n) d x_1 \cdots d x_n,
\end{equation*}
which has to hold for all compactly supported continuous functions
$G^{(n)}$, we  define  a $\sigma$-finite measure $\lambda$ on
$\mathcal{B}(\Gamma_0)$, called the {\it Lebesgue- Poisson} measure.
By (\ref{KFm}) we then  get
\begin{eqnarray}
\label{19A}
& & \int_{\Gamma} \left( \sum_{\eta \Subset \gamma} G(\eta)\right) \mu(d \gamma) =
\langle\!\langle G, k_\mu \rangle\!\rangle:= \int_{\Gamma_0} G(\eta) k_\mu(\eta) \lambda (d \eta) \\[.2cm]
& & \qquad = G^{(0)} + \sum_{n=1}^\infty \frac{1}{n!} \int_{(\mathbb{R}^d)^n} G^{(n)}
(x_1 , \dots , x_n) k^{(n)}_\mu (x_1 , \dots , x_n) d x_1 \cdots d x_n, \nonumber
\end{eqnarray}
where the first sum  runs over all finite subsets of $\gamma$.
By Lemma 2.1 in \cite{FKK} we also  have the following useful property
\begin{equation}
 \label{12A}
\int_{\Gamma_0} \sum_{\xi \subset \eta} H(\xi, \eta \setminus \xi,
\eta) \lambda (d \eta) = \int_{\Gamma_0}\int_{\Gamma_0}H(\xi, \eta,
\eta\cup \xi) \lambda (d \xi) \lambda (d \eta).
\end{equation}
Regarding the kernels in (\ref{Ra20})
we suppose that
\begin{equation}
 \label{AA}
a^{\pm} \in L^1(\mathbb{R}^d)\cap L^\infty(\mathbb{R}^d), \qquad a^{\pm}(x) = a^{\pm}(-x) \geq 0.
\end{equation}
Then we set
\begin{equation}
 \label{14A}
\langle{a}^{\pm}\rangle  =  \int_{\mathbb{R}^d} a^{\pm} (x)dx,\qquad \|a^{\pm} \|= \esssup_{x\in \mathbb{R}^d} a^{\pm}(x),
\end{equation}
and
\begin{equation}
 \label{15A}
 E^{\pm} (\eta)  =  \sum_{x\in \eta}E^{\pm} (x,\eta\setminus x) = \sum_{x\in \eta} \sum_{y\in \eta\setminus x}a^{\pm} (x-y).
\end{equation}
By (\ref{AA}), we then have
\begin{equation}
 \label{AB}
E^{\pm} (\eta) \leq  \|a^{\pm} \| |\eta|^2.
\end{equation}
`Operator' $L^\Delta_{\varepsilon,{\rm ren}}$ in (\ref{RenK}) and
(\ref{E1}) has the following structure, cf. (\ref{RenK1}),
\begin{equation}
{L}^\Delta_{\varepsilon, {\rm ren}} = A_0 +B + \varepsilon C = V + \varepsilon C \label{21A}
\end{equation}
where
\begin{gather}
(A_0 k)(\eta) = - m |\eta| k(\eta),\label{22A}\\[.2cm]
(B k)(\eta) =  - \int_{\mathbb{R}^d} E^{-} (y,\eta) k(\eta\cup y) dy + \int_{\mathbb{R}^d} \sum_{x\in \eta} a^{+} (x - y) k(\eta \setminus x \cup y) d y
, \nonumber \\[.2cm]
(C k)(\eta) = - E^{-}(\eta) k(\eta) + \sum_{x\in \eta} E^{+} (x, \eta \setminus x) k(\eta \setminus x). \nonumber
\end{gather}
Note that $L^\Delta_{1,{\rm  ren}}$ is exactly $L^\Delta$ given in \eqref{ltrian}.

\section{The Evolution of Correlation Functions}
\label{Sec4}
\subsection{The statements}
\label{CFss1}
If the competition is absent,  the correlation functions are
bounded from below by $n!$, see (\ref{99}), and hence clusters appear in the corresponding state.
A principal questions regarding the considered model
is whether the competition  can prevent from such clustering. In view of (\ref{int-d}),
the answer will be affirmative if
\[
 \|k^{(n)}\|_{L^\infty((\mathbb{R}^d)^n)} \leq C^n, \qquad n\in \mathbb{N},
\]
for some $C>0$. Then, as in eqs. (3.8) - (3.10) of \cite{Berns}, we introduce the Banach spaces of sub-Poissonian
correlation functions in which
we solve  (\ref{RenK}) and (\ref{E5}) as follows.
For $\alpha \in \mathbb{R}$, we set
\begin{equation}
  \label{z400}
 \|k\|_\alpha = \sup_{n\in \mathbb{N}_0} e^{n \alpha} \|k^{(n)}\|_{L^\infty((\mathbb{R}^d)^n)},
\end{equation}
that can also be rewritten in the form
\begin{equation}
 \label{z30}
\|k\|_\alpha = \esssup_{\eta \in \Gamma_0} |k(\eta)|\exp(\alpha|\eta|).
\end{equation}
Thereafter, we define
\begin{equation}
  \label{z300}
\mathcal{K}_\alpha = \{k:\Gamma_0\to \mathbb{R}:
\|k\|_\alpha < \infty\},
\end{equation}
which is a Banach space. In fact, we need the
scale
of such spaces $\{\mathcal{K}_\alpha: \alpha \in \mathbb{R}\}$.
For $\alpha'' <
\alpha'$, we have that $\|k\|_{\alpha''} \leq \|k\|_{\alpha'}$ ; and
hence,
\begin{equation}
 \label{z31}
\mathcal{K}_{\alpha'} \hookrightarrow \mathcal{K}_{\alpha''}, \qquad \ {\rm for} \ \alpha'' < \alpha'.
\end{equation}
Our next aim is to define $L^\Delta_{\varepsilon, {\rm ren}}$, written  in
(\ref{21A}) and (\ref{22A}), as a linear operator in $\mathcal{K}_\alpha$ for a given $\alpha\in \mathbb{R}$. Put
\begin{equation}
  \label{op-1}
\mathcal{D}_\alpha (A_0) = \{ k \in \mathcal{K}_\alpha: A_0 k \in \mathcal{K}_\alpha\}.
\end{equation}
The sets $\mathcal{D}_\alpha(B)$ and $\mathcal{D}_\alpha(C)$ are defined analogously. Then
\begin{equation}
  \label{op-2}
\mathcal{D}_\alpha (L^\Delta_{\varepsilon, {\rm  ren}}) := \mathcal{D}_\alpha (A_0) \cap \mathcal{D}_\alpha (B) \cap \mathcal{D}_\alpha (C)
\end{equation}
is the domain of $L^\Delta_{\varepsilon,{\rm ren}}$ in
$\mathcal{K}_\alpha$. Note that $\mathcal{D}_\alpha
(L^\Delta_{\varepsilon,{\rm ren}})$  is the
same for all $\varepsilon>0$, and
\begin{equation}
  \label{op-2a}
\mathcal{D}_\alpha (L^\Delta_{0,{\rm ren}}) = \mathcal{D}_\alpha(V)
= \mathcal{D}_\alpha (A_0) \cap \mathcal{D}_\alpha (B) \supset
\mathcal{D}_\alpha (L^\Delta_{\varepsilon, {\rm ren}}), \quad
\varepsilon \in (0,1].
\end{equation}
Let us show that
\begin{equation}
  \label{op-3}
\forall \alpha' > \alpha \qquad \quad \mathcal{K}_{\alpha'} \subset
\mathcal{D}_\alpha (L^\Delta_{\varepsilon, {\rm ren}}).
\end{equation}
By (\ref{z30}), we have
\begin{equation*}
|k(\eta)| \leq \|k\|_{\alpha'} \exp( - \alpha' |\eta|), \qquad \eta \in \Gamma_0.
\end{equation*}
Applying this in (\ref{22A}), by (\ref{14A}) and (\ref{AB}) we then get
\begin{equation}
  \label{op-44}
|(Ck)(\eta)| \leq |\eta|^2 \exp (- (\alpha' - \alpha) |\eta|)\left[
\| a^{-} \|  + \| a^{+} \| e^{\alpha'}\right]\|k\|_{\alpha'} \exp( -
\alpha |\eta|),
\end{equation}
and also
\begin{eqnarray}
\label{op-45}
|(Bk)(\eta)| & \leq &  \|k\|_{\alpha'} \exp( - \alpha' |\eta| - \alpha' ) \int_{\mathbb{R}^d} E^{-}(y, \eta) d y\\[.2cm]
& + & \|k\|_{\alpha'} \exp( - \alpha' |\eta|) \int_{\mathbb{R}^d} \sum_{x\in \eta} a^{+} (x-y) dy \nonumber \\[.2cm]
& \leq & |\eta|\exp (- (\alpha' - \alpha) |\eta|)\left[ \langle a^{-} \rangle e^{-\alpha'} +
\langle a^{+} \rangle\right]\|k\|_{\alpha'} \exp( - \alpha |\eta|). \nonumber
\end{eqnarray}
In a similar way, one estimates $|(A_0k)(\eta)|$.
These three estimates readily yield (\ref{op-3}).
\begin{definition}
\label{1zdf} By a classical solution of the problem (\ref{RenK}), in
the space $\mathcal{K}_\alpha$ and on the time interval $[0,T)$, we
understand a map $[0,T)\ni t \mapsto r_{t, \varepsilon} \in
\mathcal{D}_\alpha (L^\Delta_{\varepsilon, {\rm ren}})$, cf.
(\ref{op-2}), continuously differentiable on $[0,T)$, such that
(\ref{RenK}) is satisfied  for $t\in [0,T)$. A classical solution of
(\ref{E5}), cf. (\ref{op-2a}), is defined in the same way.
\end{definition}
\begin{remark}
 \label{1zrk}
In view of (\ref{op-3}), we have that $r_{t, \varepsilon} \in
\mathcal{D}_\alpha (L^\Delta_{\varepsilon, {\rm ren}})$ whenever
$r_{t, \varepsilon} \in \mathcal{K}_{\alpha_t}$ for some $\alpha_t >
\alpha$.
\end{remark}
The main assumption under which we are going to solve (\ref{RenK})
is the following: there exists $\theta >0$ such that
\begin{equation}
 \label{z14}
a^{+} (x) \leq \theta a^{-} (x), \qquad {\rm for} \ {\rm a.a.} \ x \in \mathbb{R}^d.
\end{equation}
Fix $\alpha^*\in \mathbb{R}$ and  set, cf. (\ref{14A}),
\begin{equation}
 \label{z15}
T (\alpha)  = \frac{\alpha^* - \alpha}{\langle{a}^{+}\rangle + \langle{a}^{-}\rangle e^{-\alpha}}, \qquad \alpha < \alpha^*.
\end{equation}
\begin{theorem}
 \label{2tm}
Let (\ref{z14}) be satisfied, and let $\alpha^* \in \mathbb{R}$ be such that
\begin{equation}
 \label{z16}
e^{\alpha^*} \theta < 1.
\end{equation}
Then, for each $\alpha<\alpha^*$, the problem in (\ref{RenK}) with $r_0\in \mathcal{K}_{\alpha^*}$
has a unique classical solution in $\mathcal{K}_{\alpha}$ on $[0,T (\alpha))$.
\end{theorem}
\begin{theorem}
  \label{op-tm}
For each  $\alpha^*\in \mathbb{R}$ and $\alpha<\alpha^*$, the problem in (\ref{E5}) with $r_0\in \mathcal{K}_{\alpha^*}$
has a unique classical solution in $\mathcal{K}_{\alpha}$ on $[0,T (\alpha))$.
\end{theorem}
\begin{theorem}
  \label{op1-tm}
Let the assumptions of Theorems \ref{2tm} and \ref{op-tm} be satisfied, and $r_{t, \varepsilon}$ and $r_t$ be the solution of (\ref{RenK}) and (\ref{E5}), respectively. Then, for each $\alpha < \alpha^*$ and
 $t\in (0,T (\alpha))$, it follows that
\begin{equation}
   \label{z160}
 \sup_{s\in [0,t]} \|r_{s, \varepsilon} - r_s\|_{\alpha} \to 0 , \qquad {\rm as} \ \ \varepsilon \to 0.
\end{equation}
\end{theorem}
Let us now make some comments on these statements.
The condition in (\ref{z14}) can certainly be satisfied if the dispersal kernel decays
faster than the competition kernel. The magnitude parameter $\theta$ determines the large $n$ asymptotics of the initial
correlation function, see (\ref{z400}) and (\ref{z16}). Note that in Theorem \ref{op-tm} we do not require
(\ref{z14}).
 The main characteristic feature of the solutions mentioned in Theorems \ref{2tm} and \ref{op-tm} is that, at a given $t$, they lie
in a  space, $\mathcal{K}_\alpha$, `bigger' than the initial $r_0$ does, cf. (\ref{z31}).  The bigger $t$, the bigger should be the space $\mathcal{K}_\alpha$.
 The function $(- \infty, \alpha^*) \ni \alpha \mapsto T(\alpha)$ defined in (\ref{z15}) is bounded from above by a certain $T^* (\langle a^+ \rangle, \langle a^- \rangle, \alpha^*)$ beyond which the solutions of both problems  could not be extended, see, however, Remark \ref{Decrk} below.

\subsection{The proof of the statements}
\label{OvcS}

The proof is based on three lemmas formulated  below.
\begin{lemma}
  \label{op-lm}
Let $\theta$ be as in (\ref{z14}) and $\alpha \in \mathbb{R}$ be such that $e^{\alpha} \theta <1$, cf. (\ref{z16}). Then, for arbitrary $\varepsilon \in (0,1]$,
there exists a closed subspace, $\mathcal{A}_{\alpha , \varepsilon}\subset \mathcal{K}_\alpha$, and a $C_0$-semigroup of linear contractions, $S_{\alpha, \varepsilon}(t): \mathcal{A}_{\alpha, \varepsilon} \to \mathcal{A}_{\alpha, \varepsilon}$, $t\geq 0$, with generator $A_{\alpha, \varepsilon}$, such that, for each $\alpha' > \alpha$, cf. (\ref{z31}), the following holds:
\vskip.1cm
\begin{itemize}
\item[(a)] $\mathcal{K}_{\alpha'} \subset {\rm Dom} (  A_{\alpha, \varepsilon}) \subset \mathcal{A}_{\alpha, \varepsilon}$;
\item[(b)] for each $k\in \mathcal{K}_{\alpha'}$, $A_{\alpha, \varepsilon} k = (A_0 + \varepsilon C) k$, where the latter two operators are defined in (\ref{22A}) and (\ref{op-1}), (\ref{op-2});
\item[(c)] for each $\alpha''<\alpha$, cf. (\ref{z31}), the restriction of $S_{\alpha'', \varepsilon}$ to
$\mathcal{A}_{\alpha, \varepsilon}$ coincides with $S_{\alpha, \varepsilon}$.
\end{itemize}
\end{lemma}
 \begin{lemma}
\label{op-lm2}
For $\varepsilon =0$, all the statements of Lemma \ref{op-lm} hold true without any restrictions on $\alpha$. Furthermore, for each $\alpha' >\alpha$, $\alpha'' \leq \alpha$, $\varepsilon \in (0,1]$,  and any $t>0$, it follows that
\begin{equation}
  \label{z162}
\sup_{s\in [0,t]}\|(S_{\alpha'', \varepsilon}(s) - S_{\alpha'', 0}(s))\|_{\alpha'\alpha} \leq \varepsilon t M(\alpha' - \alpha),
\end{equation}
with $M(\varkappa)$  given by the formula
\begin{equation}
  \label{Mkappa}
  M(\varkappa) := \left(\frac{2}{e \varkappa}\right)^2 \left( \|a^{-} \| + \|a^{+} \|e^{\alpha^*} \right).
\end{equation}
\end{lemma}
\begin{corollary}
 \label{1pn}
For each $\alpha' >\alpha$, $\varepsilon \in [0,1]$, and  any $k\in \mathcal{K}_{\alpha'}$, the map
\begin{equation*}
[0, +\infty) \ni t \mapsto S_{\alpha, \varepsilon}(t) k \in \mathcal{K}_\alpha
\end{equation*}
is continuous.
\end{corollary}
\begin{remark}
  \label{op-lm-rk}
By claim (c) of Lemma \ref{op-lm}, for $\alpha' > \alpha$ and $\alpha'' \leq \alpha$, each $S_{\alpha'', \varepsilon}(t)$ can be considered as a bounded linear contraction from $\mathcal{K}_{\alpha'}$ to $\mathcal{K}_\alpha$.
\end{remark}
\begin{myproof}{Theorem \ref{2tm}}
For $\alpha' > \alpha$ and a bounded linear operator $Q: \mathcal{K}_{\alpha'} \to \mathcal{K}_\alpha$,
cf. (\ref{z31}), by $\|Q\|_{\alpha'\alpha}$ we denote the corresponding operator norm.
For $\varkappa>0$ and $\eta \in \Gamma_0$, we have
\begin{equation}
  \label{z11Aa}
  |\eta| e^{-\varkappa |\eta|} \leq \frac{1}{ e \varkappa}, \qquad \quad |\eta|^2 e^{-\varkappa |\eta|} \leq \left(\frac{2}{ e \varkappa}\right)^2.
\end{equation}
Let us now fix some $\alpha_* < \alpha^*$ and then set $T_* = T(\alpha_*)$.
For $\alpha', \alpha \in [\alpha_*, \alpha^*]$ such that $\alpha' > \alpha$, the expression for $B$ given in (\ref{22A}) can be used to define a bounded linear operator
$B:\mathcal{K}_{\alpha' } \to \mathcal{K}_\alpha$. We shall keep the notation $B$ for this operator if it is clear between which spaces it acts.
However, additional labels will be used if we deal with more than one such $B$  acting between different spaces.
The norm of $B:\mathcal{K}_{\alpha' } \to \mathcal{K}_\alpha$ can be estimated by means of (\ref{op-45}) and (\ref{z11Aa}) as follows, cf. (\ref{z15}),
\begin{equation}
  \label{nB}
\|B\|_{\alpha'\alpha} \leq \frac{\alpha^* - \alpha_*}{e(\alpha' - \alpha)T_*} .
\end{equation}
In a similar way, we define also bounded operators $A_0, C:\mathcal{K}_{\alpha' } \to \mathcal{K}_\alpha$, for which we get, cf. (\ref{22A}),  (\ref{op-44}), (\ref{Mkappa}), and (\ref{z11Aa}),
\begin{equation}
  \label{op-50}
 \|A_0 \|_{\alpha' \alpha} \leq \frac{m}{e(\alpha' - \alpha)}, \qquad  \|C \|_{\alpha' \alpha} \leq \left(\frac{2 }{e(\alpha' - \alpha)}\right)^2 \left(\|a^{-} \|+ \|a^{+}\|e^{\alpha^*} \right).
\end{equation}
Now we fix $\alpha \in (\alpha_*, \alpha^*)$ and recall that $T(\alpha)$ is defined in (\ref{z15}). Then,
for $r_0\in\mathcal{K}_{\alpha^{*}}$ as in (\ref{RenK}) and $t\in [0,T(\alpha))$,
we define: $u_{t,0} = 0$, and, for  $n \in \mathbb{N}$,
\begin{equation}
  \label{nB1}
u_{t,n}  =  S_{\alpha, \varepsilon} (t) r_0  +  \sum_{l=1}^{n-1} \int_0^t
\int_0^{t_1} \cdots \int_0^{t_{l-1}} U_{\alpha, \varepsilon}^{(l)} (t,t_1, \dots, t_l) r_0 d t_l \cdots dt_1,
\end{equation}
where $t_0=t$ and
\begin{eqnarray}
  \label{nBB1}
U_{\alpha, \varepsilon}^{(l)} (t,t_1, \dots, t_l) & := & S_{\alpha, \varepsilon} (t-t_1)B_1 S_{\alpha, \varepsilon} (t_1 - t_2) B_{2} \\[.2cm] & & \qquad \qquad \qquad \times   \cdots \times   S_{\alpha, \varepsilon} (t_{l-1}-t_l)B_l
 S_{\alpha, \varepsilon} (t_l). \nonumber
\end{eqnarray}
The latter operator is suppose to act from $\mathcal{K}_{\alpha^*}$ to $\mathcal{K}_{\alpha}$. In order to define the action of the operators in
the product,
we fix $q>1$ such that $q t< T (\alpha)$ and introduce
\begin{equation}
  \label{nBB2}
 \alpha_{2p} = \alpha_0 - p \epsilon_1 - p\epsilon_2 , \quad \alpha_{2p+1} = \alpha_0 - (p+1) \epsilon_1 - p\epsilon_2, \quad p= 0, \dots, l,
\end{equation}
where $\alpha_0 = \alpha^*$,  $\alpha_{2l+1} = \alpha$, and
\begin{equation}
  \label{nBB3}
\epsilon_1 = \frac{(q-1) (\alpha^* - \alpha)}{q(l+1)}, \qquad \epsilon_2 = \frac{\alpha^* - \alpha}{ql} .
\end{equation}
Then the operators in the product in (\ref{nBB1}) are set to act as follows, cf. Lemma \ref{op-lm}, (\ref{22A}), and (\ref{nB}),
\begin{eqnarray}
  \label{nBB4}
& & S_{\alpha, \varepsilon}(t_l): \mathcal{K}_{\alpha_0 } \to \mathcal{K}_{\alpha_1 }, \\[.2cm]
& & S_{\alpha, \varepsilon}(t_{l-p}-t_{l-p+1}): \mathcal{K}_{\alpha_{2p} } \to \mathcal{K}_{\alpha_{2p+1} }, \nonumber\\[.2cm]
& & B_{l-p+1} : \mathcal{K}_{\alpha_{2p-1} } \to \mathcal{K}_{\alpha_{2p} }, \qquad p = 1, \dots l. \nonumber
\end{eqnarray}
Since all $S_{\alpha, \varepsilon}(t)$ are contractions, from  (\ref{nB1}) we have
\begin{equation}
  \label{nB2}
\|u_{t, n+1} - u_{t, n}\|_{\alpha} \leq \frac{t^n}{n!} \|B_1 B_{2} \cdots B_n\|_{\alpha^* \alpha} \|r_0\|_{\alpha^*}, \qquad n\in \mathbb{N}.
\end{equation}
Note that the right-hand side of (\ref{nB2}) is independent of $\varepsilon$.
To estimate it we use (\ref{nB}), then (\ref{nBB3}) and the last line in (\ref{nBB4}) with $l=n$, which yields, cf. (\ref{nBB3}),
\begin{equation}
  \label{Bkappa}
\|B_{n-p+1}\|_{\alpha_{2p-1 } \alpha_{2p}} \leq \frac{\alpha^* - \alpha}{e\epsilon_2 T (\alpha)} = \frac{qn}{e  T(\alpha)}, \qquad p = 1, \dots n.
\end{equation}
Applying this in (\ref{nB2}) we obtain
\begin{equation}
  \label{Bkapp}
\sup_{s\in [0,t]}\|u_{s, n+1} - u_{s, n}\|_{\alpha} \leq \frac{1}{n!} \left(\frac{n}{e}\right)^n \left(\frac{ qt}{T(\alpha)}\right)^n,
\end{equation}
which means that $\{u_{s, n}\}_{n\in \mathbb{N}_0}$ is a Cauchy
sequence in each of  $\mathcal{K}_{\alpha'}$, $\alpha'\in [\alpha_*,
\alpha]$, uniformly in $\varepsilon \in (0,1]$ and in  $s\in [0,t]$.
Let $u_s\in \mathcal{K}_{\alpha}$ be its limit. By Corollary
\ref{1pn}, $U_{\alpha, \varepsilon}^{(l)}$ is continuous in $t_1,
\dots t_l$, and hence $u_{t,n} \in \mathcal{K}_{\alpha}$ is
continuously differentiable in $\mathcal{K}_{\alpha}$ on  $t\in [0,T
(\alpha))$. Then, see claim (b) of Lemma \ref{op-lm},
\begin{equation}
 \label{nBa}
 \frac{d}{dt} u_{t,n} = A_{\alpha, \varepsilon} u_{t,n} + B u_{t,n-1} =  (A_0 + \varepsilon C)u_{t,n} + B u_{t,n-1},
\end{equation}
where both $(A_0 + \varepsilon C)$ and $B$ act from
$\mathcal{K}_\alpha$ to $\mathcal{K}_{\alpha_*}$. On the other hand,
\begin{multline*}
\sup_{s\in [0,t]}\|(A_0 + \varepsilon C)(u_{s, n+1} - u_{s, n})\|_{\alpha_*} \\\leq \bigg{(}\|A_0\|_{\alpha \alpha_*} + \|C\|_{\alpha \alpha_*}
\bigg{)} \sup_{s\in [0,t]}\|u_{s, n+1} - u_{s, n}\|_{\alpha},
\end{multline*}
\[\sup_{s\in [0,t]}\|B(u_{s, n+1} - u_{s, n})\|_{\alpha_*} \leq \|B\|_{\alpha \alpha_*} \sup_{s\in [0,t]}\|u_{s, n+1} - u_{s, n}\|_{\alpha},
\]
where the operator norms can be estimated as in (\ref{op-50}).
Hence, by (\ref{Bkapp}) $\{{d}u_{s,n} /{ds}  \}_{n\in \mathbb{N}}$
is a Cauchy sequence in $\mathcal{K}_{\alpha_*}$, uniformly in $s\in
[0,t]$. Therefore,  the limiting $u_s \in \mathcal{K}_\alpha \subset
\mathcal{K}_{\alpha_*}$ is continuously differentiable on $[0,t]$,
and
\[
{d}u_{s,n} /{ds} \to {d}u_{s} /{ds}, \qquad n \to +\infty.
\]
On the other hand, the right-hand side of (\ref{nBa}) converges in
$\mathcal{K}_{\alpha_*}$ to ${L}^\Delta_{\varepsilon, {\rm ren}}
u_s$, see (\ref{21A}). Hence, $u_s$ is the classical solution on
$[0,t]$, see Definition \ref{1zdf} and Remark \ref{1zrk}. Since such $t$ can be arbitrary in $(0,T_*)$, this yields
the existence of the solution in question.

Let us now prove the uniqueness. In view of the linearity of the
problem in (\ref{RenK}), it is enough to prove that the zero function
is its only solution corresponding to the zero initial condition.
Let $\alpha \in (\alpha_*, \alpha^*)$ be fixed, and let $t\in (0,
T_*)$ be such that $t< T(\alpha)$, see (\ref{z15}). Then the problem
in (\ref{RenK}) with the zero initial condition has classical
solutions in $\mathcal{K}_{\alpha}$ on $[0, t]$. Each of them solves
also the following integral equation, cf. (\ref{nB1}),
\begin{equation}
  \label{NOV}
u_t = \int_0^t S_{\alpha, \varepsilon}(t-s) B u_s ds,
\end{equation}
see the proof of Theorem IX.1.19, page 486 in \cite{Kato}. Since $t< T(\alpha)$, there exists $\alpha' \in (\alpha_*, \alpha)$ such that
also $t < T(\alpha')$. Note that
$u_t$ in the left-hand side of (\ref{NOV}) lies in $\mathcal{K}_{\alpha'}$, cf. (\ref{z31}). Then the meaning of (\ref{NOV}) is the following. As a bounded operator, cf (\ref{nB}), $B$ maps $u_s \in \mathcal{K}_{\alpha}$ to $Bu_s \in \mathcal{K}_{\alpha''}$, $\alpha'' \in (\alpha', \alpha)$, continuously in $s$. Furthermore, in view of claim (c) of Lemma \ref{op-lm}, $S_{\alpha, \varepsilon}$ in (\ref{NOV}) can be replaced by $S_{\alpha', \varepsilon}$. Then, by Corollary \ref{1pn}, the integrand in (\ref{NOV}) is continuous in $s$, and hence the integral makes sense.
Now we iterate (\ref{NOV}) $n$ times and get, cf. (\ref{nB1}) and (\ref{nBB1}),
\begin{eqnarray*}
 u_t & = & \int_0^t S_{\alpha', \varepsilon}(t-t_1) B_1  \int_0^{t_1} S_{\alpha', \varepsilon}(t_1- t_2) B_{2} \times \\[.2cm]
 & \times & \cdots \times  \int_0^{t_n} S_{\alpha', \varepsilon}(t_{n-1}- t_n) B_n u_{t_n} d t_n \cdots dt_1,
  \nonumber
\end{eqnarray*}
from which we deduce the following estimate, cf (\ref{nB2}),
\begin{equation}
  \label{NOV3}
\|u_t\|_{\alpha'} \leq \frac{t^n}{n!} \|B_1 B_{2} \cdots B_n\|_{\alpha\alpha'} \sup_{s\in [0,t]}\|u_s \|_{\alpha}.
\end{equation}
Recall that $[0,t]\ni s \mapsto u_s \in \mathcal{K}_{\alpha}$ is continuous. Similarly as in the case of (\ref{nB2}) we then get
\[
\|B_1 B_{2} \cdots B_n\|_{\alpha\alpha'} \leq \left(\frac{n}{e}\right)^n \left[ \frac{\alpha^* - \alpha}{T(\alpha) (\alpha - \alpha')}\right]^n.
\]
Hence, by picking large enough $n$ the right-hand side of (\ref{NOV3}) can be made arbitrarily small whenever
\[
t < \left(\frac{\alpha - \alpha'}{\alpha^* - \alpha}\right) T(\alpha).
\]
This proves that $u_t$ is the zero element of $\mathcal{K}_{\alpha'}$ for small enough $t>0$. Since $u_t$
lies  in $\mathcal{K}_{\alpha}\subset \mathcal{K}_{\alpha'}$, then $u_t$ is also the zero element of $\mathcal{K}_{\alpha}$.
Now we repeat the above procedure on the interval $[t, 2t]$ and obtain that $u_{2t}$ is also the zero element of $\mathcal{K}_{\alpha}$.
The extension of this fact to each $t< T_*$  then follows.
\end{myproof}
\begin{remark}
  \label{Ovcrk}
The proof of the uniqueness has been done in the spirit of Ovcyannikov's method,
see pages 16 and 17 in \cite{trev}. Here, however, we deal with operators which
cannot be directly accommodated to the Ovcyannikov scheme, see the second estimate in (\ref{op-50}).
We have managed to take these operators into account through the `sun-dual' semigroup mentioned in Lemma \ref{op-lm}.
\end{remark}
\begin{myproof}{Theorem \ref{op-tm}}
By the same procedure as in the proof of Theorem \ref{2tm} we obtain that the sequence defined recursively, cf. (\ref{nB1}),  by
\begin{eqnarray}
  \label{nB100}
v_{t,n}  =  S_{\alpha_*,0} (t) r_0 & + & \sum_{l=1}^{n-1} \int_0^t  \int_0^{t_1} \cdots \int_0^{t_{l-1}} S_{\alpha_*, 0} (t-t_1)B_1 \\[.2cm]
& \times & S_{\alpha_*, 0} (t_1-t_2)B_{2} \cdots S_{\alpha_*, 0} (t_{l-1}-t_l)B_l S_{\alpha_*, 0}(t_l) r_0 d t_l \cdots dt_1, \nonumber
\end{eqnarray}
converges in $\mathcal{K}_{\alpha_*}$  to a classical solution of (\ref{E5}). In contrast to the case of $\varepsilon >0$,
the semigroup $S_{\alpha, 0}$ can be constructed explicitly: cf. claim (b) of Lemma \ref{op-lm},
\begin{equation}
\label{nB101}
(S_{\alpha, 0}(t) v)(\eta) = \exp(-t m |\eta|) v(\eta), \qquad v \in \mathcal{K}_\alpha,
\end{equation}
which allows for omitting (\ref{z16}) in this case. The rest of the proof goes  exactly as in  Theorem \ref{2tm}.
\end{myproof}
\begin{myproof}{Theorem \ref{op1-tm}}
Let $r_{t, \varepsilon}$ be the limit of the sequence $\{u_{t,n}\}_{n\in \mathbb{N}}$ defined in (\ref{nB1}). We prove that this $r_{t, \varepsilon}$ converges to $r_t$, as stated in the theorem.
Given $\delta>0$ and $t\in (0,T_*)$, let $n_\delta\in \mathbb{N}$ be such that, for all $n> n_\delta$, both following estimates hold
\begin{equation}
  \label{nB200}
\sup_{s\in [0,t]} \|u_{s, n} - r_{s, \varepsilon}\|_{\alpha_*} < \delta/3, \qquad  \sup_{s\in [0,t]} \|v_{s, n} - r_{s}\|_{\alpha_*} < \delta/3,
\end{equation}
where $v_{s,n}$ is defined in (\ref{nB100}). Note that
the first estimate in (\ref{nB200}) is uniform in $\varepsilon\in (0,1]$ as the right-hand side of (\ref{nB2}) is $\varepsilon$-independent.
By (\ref{nB1}) and (\ref{nB100}), we get
\begin{eqnarray}
  \label{nB201}
 u_{s,n} - v_{s,n} & = & Q_{ \varepsilon}(s) r_0\\[.2cm] &  + & \sum_{l=1}^{n-1} \int_{0}^s \int_{0}^{s_1} \cdots \int_0^{s_{l-1}} \sum_{p=1}^{l+1}
R^{(p,l)}_{\varepsilon}(s, s_1, \dots, s_l)r_0 ds_l \cdots ds_1, \nonumber
\end{eqnarray}
where
\begin{equation*}
Q_{\varepsilon}(s):= S_{\alpha_*, \varepsilon}(s) - S_{\alpha_*, 0}(s) ,
\end{equation*}
\begin{equation}
  \label{nB203}
R^{(0,l)}_{\varepsilon}(s, s_1, \dots, s_l):= S_{\alpha_*, \varepsilon}(s-s_1) B_1 \cdots  S_{\alpha_*, \varepsilon}(s_{l-1}-s_l)B_{l} Q_{ \varepsilon}(s_l),
\end{equation}
and, for $p=1, \dots , l$,
\begin{eqnarray}
  \label{nB204}
& & R^{(p,l)}_{\varepsilon}(s, s_1, \dots, s_l):=  S_{\alpha_*, \varepsilon}(s-s_1) B_1  \cdots S_{\alpha_*, \varepsilon}(s_{l-p-1}-s_{l-p}) B_{l-p} \\[.2cm] & &  \qquad \qquad \times Q_{\varepsilon}(s_{l-p}-s_{l-p+1}) B_{l-p+1} S_{\alpha_*, 0}(s_{l-p+1}-s_{l-p+2})B_{l-p+1} \nonumber \\[.2cm] & & \qquad \qquad  \times \cdots \times S_{\alpha_*, 0}(s_{l-1}-s_l)B_l S_{\alpha_*, 0}(s_l).\nonumber
\end{eqnarray}
In (\ref{nB204}), the operators $S_{\alpha_*, \varepsilon}$ and $B_p$ act as in (\ref{nBB4}). Then taking into account that $\|S_{\alpha_*, \varepsilon}(s)\|_{\alpha'\alpha} \leq 1$, see Lemma \ref{op-lm}, and likewise $\|S_{\alpha_*, 0}(s)\|_{\alpha'\alpha} \leq 1$, cf. (\ref{nB101}), for a fixed $s\in (0,T_*)$ and $q>1$ such that $q s< T_*$, we obtain from (\ref{nB203}), (\ref{nB204}) and from (\ref{z162}), (\ref{Mkappa}), (\ref{nBB2}), (\ref{nBB3}), (\ref{nBB4}), and (\ref{Bkappa}) that
\begin{gather*}
 \|R^{(0,l)}_{\varepsilon}(s, s_1, \dots, s_l)\|_{\alpha^* \alpha_*} \leq \varepsilon K (l+1)^2 s_l \|B_1\cdots B_l\|_{\alpha^* \alpha_*} \\ \leq \varepsilon K (l+1)^2 s_l \left(\frac{ql}{e T_*}\right)^l,
\end{gather*}
and likewise for $p=1, \dots, l$,
\begin{equation*}
  \|R^{(p,l)}_{\varepsilon}(s, s_1, \dots, s_l)\|_{\alpha^* \alpha_*} \leq \varepsilon K (l+1)^2 (s_{l-p} - s_{l-p+1}) \left(\frac{ql}{e T_*}\right)^l,
\end{equation*}
where
\[
K := \left[ \frac{2 q}{e(q-1)(\alpha_* - \alpha_*)  }\right]^2.
\]
Applying both latter estimates in (\ref{nB201}) we finally get that, for $t\in [0,T_*)$ and $q>1$ such that $qt<T_*$ , the following holds
\begin{equation}
  \label{z401}
\sup_{s\in [0,t]}\|u_{s,n} - v_{s,n}\|_{\alpha_*} \leq \varepsilon \|r_0\|_{\alpha^*} \varphi_q (t),
\end{equation}
where
\[
\varphi_q(t) := t K \sum_{l=0}^\infty \frac{1}{l!} \left(\frac{l}{e}\right)^l (l+1)^2 \left(\frac{qt}{T_*}\right)^l.
\]
Now we fix $n>n_\delta$ such that both estimates in (\ref{nB200}) hold, independently of $\varepsilon$. Next, for this fixed $n$, we pick $\varepsilon$ such that also the left-hand side of (\ref{z401}) is less than $\delta/3$, which by the triangle inequality yields (\ref{z160}).
\end{myproof}

\subsection{The proof of Lemmas \ref{op-lm} and \ref{op-lm2}}
\label{op-lmS}
The semigroups in question will be obtained with the help of the corresponding semigroups constructed in the pre-dual spaces.

For a given $\alpha\in \mathbb{R}$, the space $\mathcal{K}_\alpha$ defined in (\ref{z300}), (\ref{z30}) is dual to the Banach space
$\mathcal{G}_\alpha := L^1 (\Gamma_0 , e^{-\alpha |\cdot|} d \lambda)$,
with norm
\begin{equation*}
\|G\|_{\alpha} = \int_{\Gamma_0} |G(\eta)| \exp(- \alpha |\eta|) \lambda (d \eta)= \sum_{n=0}^\infty \frac{1}{n!} e^{-n \alpha} \|G^{(n)}\|_{L^1((\mathbb{R}^d)^n)}.
\end{equation*}
 The duality is defined by the pairing $\langle \! \langle G, k \rangle \! \rangle$ as in (\ref{19A}).
For $G: \Gamma_0 \to \mathbb{R}$, we set
\begin{gather}\label{z111}
(\widehat{A}^{(1)}_{\varepsilon} G) (\eta) =   -(|\eta|m + \varepsilon E^{-}(\eta)) G(\eta):= - E(\eta)G(\eta), \\[.2cm] (\widehat{A}^{(2)}_{\varepsilon} G) (\eta) =   \varepsilon\int_{\mathbb{R}^d} E^{+} (y,\eta) G (\eta\cup y) dy, \nonumber
\end{gather}
and
\begin{equation}
  \label{z1}
  \widehat{A}_\varepsilon  = \widehat{A}^{(1)}_\varepsilon + \widehat{A}^{(2)}_\varepsilon.
\end{equation}
For $\alpha$ as in Lemma \ref{op-lm}, we define
\begin{eqnarray}
 \label{z6}
\mathcal{D}_\alpha(\widehat{A}^{(1)}_\varepsilon)& = & \{G\in \mathcal{G}_\alpha : E(\cdot) G(\cdot) \in \mathcal{G}_\alpha\},\\[.2cm]
\mathcal{D}_\alpha(\widehat{A}^{(2)}_\varepsilon)& = & \{G\in \mathcal{G}_\alpha : E^{+}(\cdot) G(\cdot) \in \mathcal{G}_\alpha\},\nonumber
\end{eqnarray}
where $E^{\pm}(\eta)$ are given in (\ref{15A}). Now we use (\ref{z111}) to define the corresponding operators in $\mathcal{G}_\alpha$.
As a multiplication operator, $\widehat{A}^{(1)}_\varepsilon: \mathcal{G}_\alpha \to \mathcal{G}_\alpha$ with ${\rm Dom}(\widehat{A}^{(1)}_\varepsilon) = \mathcal{D}_\alpha(\widehat{A}^{(1)}_\varepsilon)$ is closed. By (\ref{12A}), for  $G\in \mathcal{D}_\alpha(\widehat{A}^{(2)}_\varepsilon)$, we get
\begin{eqnarray}
 \label{z7}
\|\widehat{A}^{(2)}_\varepsilon G\|_\alpha & \leq & \varepsilon \int_{\Gamma_0} \int_{\mathbb{R}^d} E^{+}(y, \eta) |G(\eta\cup y)|e^{-\alpha|\eta|}dy \lambda(d\eta)\\[.2cm]
& = & e^\alpha \varepsilon \int_{\Gamma_0}  |G(\eta)| e^{-\alpha|\eta|} \left(\sum_{x\in \eta} E^{+} (x, \eta \setminus x) \right) \lambda ( d \eta)
\nonumber\\[.2cm]
& = & e^\alpha\varepsilon  \|E^{+} (\cdot) G(\cdot)\|_\alpha .\nonumber
\end{eqnarray}
Hence,  $\widehat{A}^{(2)}_\varepsilon: \mathcal{G}_\alpha \to \mathcal{G}_\alpha$ with ${\rm Dom}(\widehat{A}^{(2)}_\varepsilon) = \mathcal{D}_\alpha(\widehat{A}^{(2)}_\varepsilon)$ is well-defined.

We say that $G\in \mathcal{G}_\alpha$ is positive if $G(\eta)\geq 0$ for $\lambda$-almost all $\eta \in \Gamma_0$.
Let $\mathcal{G}_\alpha^{+}$ be the cone of all positive
$G\in \mathcal{G}_\alpha$. An operator $(Q, {\rm Dom} Q)$ on
$\mathcal{G}_\alpha$
is said to be positive if $Q: {\rm Dom}( Q )\cap \mathcal{G}_\alpha^{+} \to \mathcal{G}_\alpha^{+}$.
A semigroup of operators $S(t): \mathcal{G}_\alpha \to \mathcal{G}_\alpha$, $t\geq 0$, is called {\it sub-stochastic} if each $S(t)$ is positive and
$\|S(t) G\|_{\alpha} \leq \|G\|_\alpha$ for all $G\in \mathcal{G}_\alpha$.
The proof of Lemma \ref{op-lm} is based on Theorem 2.2 of \cite{TV} which we formulate here
in the following form.
\begin{proposition}
  \label{1lm}
Let $(Q_0, {\rm Dom}(Q_0))$ be the generator of a $C_0$-semigroup of positive operators on $\mathcal{G}_\alpha$, and let $Q_1 : {\rm Dom}(Q_0)\to \mathcal{G}_\alpha$ be positive and such that, for all
$G\in {\rm Dom} (Q_0) \cap \mathcal{G}_\alpha^{+}$,
\begin{equation*}
\int_{\Gamma_0} ((Q_0 + Q_1)G)(\eta) \exp(- \alpha |\eta|) \lambda (d \eta) \leq 0.
\end{equation*}
Then, for all $\varkappa \in (0,1)$, the operator $(Q_0 + \varkappa Q_1, {\rm Dom}(Q_0))$ is the generator of a sub-stochastic semigroup on
$\mathcal{G}_\alpha$.
\end{proposition}
\begin{myproof}{Lemma \ref{op-lm}}
The operator $\widehat{A}_\varepsilon^{(1)}$ defined in (\ref{z111}) and (\ref{z6}) generates a positive semigroup on $\mathcal{G}_\alpha$. For  a  $\varkappa \in (0,1)$, by (\ref{12A}) we have, cf. (\ref{z7}),
\begin{eqnarray}
\label{z701}
& & \int_{\Gamma_0} \left((\widehat{A}_\varepsilon^{(1)} + \varkappa^{-1} \widehat{A}_\varepsilon^{(2)}  )G(\eta)\right)\exp(-\alpha |\eta|) \lambda (d \eta)\\[.2cm]
& & \quad = - m \int_{\Gamma_0} |\eta| G(\eta) \exp(-\alpha |\eta|) \lambda (d \eta) - \varepsilon \int_{\Gamma_0} E^{-}(\eta) G(\eta) \exp(-\alpha |\eta|) \lambda (d \eta)  \nonumber \\[.2cm]
& & \quad + \frac{\varepsilon}{\varkappa}\int_{\Gamma_0} \int_{\mathbb{R}^d} E^{+} (y, \eta) G(\eta \cup y) \exp(-\alpha |\eta|) \lambda (d \eta) \nonumber \\[.2cm]
& & \quad = - m \int_{\Gamma_0} |\eta| G(\eta) \exp(-\alpha |\eta|) \lambda (d \eta)  \nonumber \\[.2cm]
& & \quad - \varepsilon \int_{\Gamma_0} \left[E^{-} (\eta) - \varkappa^{-1} e^{\alpha} E^{+} (\eta) \right]G(\eta)\exp(-\alpha |\eta|) \lambda (d \eta).
\nonumber \end{eqnarray}
Since $e^\alpha \theta <1$, we can pick $\varkappa\in (0,1)$ such that also $e^\alpha \theta\varkappa^{-1} <1$. For this $\varkappa$, by (\ref{z14}) we have $\left[E^{-} (\eta) - \varkappa^{-1} e^{\alpha} E^{+} (\eta) \right] \geq 0$ for $\lambda$-almost all $\eta$, which means that the left-hand side of
(\ref{z701}) is non-positive for such $\varkappa$.  We also have that $\varkappa^{-1}\widehat{A}_\varepsilon^{(2)}$ is  positive and defined on the domain of $\widehat{A}_\varepsilon^{(1)}$ by (\ref{z7}). Then, by Proposition \ref{1lm},  the operator $\widehat{A}_\varepsilon = \widehat{A}_\varepsilon^{(1)} + \varkappa (\varkappa^{-1} \widehat{A}_\varepsilon^{(2)})$, cf. (\ref{z1}), is the generator of a sub-stochastic semigroup on $\mathcal{G}_\alpha$, which we denote by $\widehat{S}_{\alpha, \varepsilon}(t)$, $t\geq 0$. Note that this in particular means
\begin{equation}
  \label{z1600}
 \|\widehat{S}_{\alpha, \varepsilon}(t) G\|_\alpha \leq \|G\|_\alpha, \qquad  t >0, \quad  G \in \mathcal{G}_\alpha.
\end{equation}
For a fixed $t>0$, let $\widehat{S}^*_{\alpha, \varepsilon}(t)$ be the operator adjoint  to $\widehat{S}_{\alpha, \varepsilon}(t)$.
All such operators constitute a semigroup on $\mathcal{K}_\alpha$, which, however is not strongly continuous as the space $\mathcal{K}_\alpha$
is of $L^\infty$-type. Let
$\widehat{A}_\varepsilon^*$ be the adjoint to $\widehat{A}_\varepsilon$. Its domain is, cf. (\ref{19A}),
\begin{equation}
 \label{z32}
{\rm Dom}(\widehat{A}_\varepsilon^*) =\bigl\{k\in \mathcal{K}_\alpha:  \exists
\tilde{k}\in \mathcal{K}_\alpha \ \  \forall G\in {\rm Dom}(\widehat{A}_\varepsilon)  \ \ \langle\! \langle \widehat{A}_\varepsilon G, k\rangle \!\rangle =
\langle\! \langle G, \tilde{k}\rangle\!\rangle\bigr\}.
\end{equation}
By $\mathcal{A}_{\alpha, \varepsilon}$ we denote the closure of (\ref{z32}) in $\mathcal{K}_\alpha$.
From the very definition in (\ref{z32}) by (\ref{12A})
it follows that
\begin{equation}
  \label{z702}
\forall \alpha'> \alpha\qquad \quad \mathcal{K}_{\alpha'} \subset {\rm Dom}(\widehat{A}_\varepsilon^*) \subset \overline{{\rm Dom}(\widehat{A}_\varepsilon^*)}  =:\mathcal{A}_{\alpha, \varepsilon},
\end{equation}
and
\begin{equation}
  \label{z703}
 \forall k \in \mathcal{K}_{\alpha'} \qquad \quad \widehat{A}_\varepsilon^* k= (A_0 + \varepsilon C) k.
\end{equation}
Note that $\mathcal{A}_{\alpha, \varepsilon}$ is a proper subspace of $\mathcal{K}_\alpha$. Now, for $t>0$, we set
\begin{equation}
  \label{z704}
  S_{\alpha, \varepsilon}(t) = \widehat{S}^*_{\alpha, \varepsilon}(t)\left\vert_{\mathcal{A}_{\alpha, \varepsilon}} \right..
\end{equation}
By Theorem 10.4, page 39 in \cite{Pazy},
the collection $\{S_{\alpha, \varepsilon}(t)\}_{t\geq 0}$ constitutes a $C_0$-semigroup on $\mathcal{A}_{\alpha, \varepsilon}$, called sometimes `sun-dual' to $\{\widehat{S}_{\alpha, \varepsilon}(t)\}_{t\geq 0}$. Its generator
$A_{\alpha, \varepsilon}$ is the part of $\widehat{A}_\varepsilon^*$ in $\mathcal{A}_{\alpha, \varepsilon}$, that is, the restriction of
$\widehat{A}^*_{\alpha, \varepsilon}$
to the set
\begin{equation}
 \label{z34a}
{\rm Dom} (A_{\alpha, \varepsilon}):= \{ k\in {\rm Dom}(\widehat{A}_\varepsilon^*): \widehat{A}_\varepsilon^*k \in \mathcal{A}_{\alpha, \varepsilon}\}.
\end{equation}
By (\ref{op-44}) and (\ref{op-45}), it can be shown that, for any $\alpha'' \in (\alpha ,\alpha')$, both $A_0$ and $C$ act as bounded operators from $\mathcal{K}_{\alpha'}$ to $\mathcal{K}_{\alpha''}$. Therefore, $(A_0 + \varepsilon C)k \in \mathcal{K}_{\alpha''} \subset \mathcal{A}_{\alpha , \varepsilon}$, and hence
\begin{equation}
  \label{z34b}
  \mathcal{K}_{\alpha'} \subset {\rm Dom} (A_{\alpha, \varepsilon}).
\end{equation}
Thus, the objects introduced in (\ref{z703}), (\ref{z704}), and (\ref{z34a}) have the properties stated in the lemma,
cf. (\ref{z702}) and (\ref{z34b}).
\end{myproof}
\vskip.1cm
\begin{myproof}{Lemma \ref{op-lm2}}
For $\varepsilon =0$, we have, cf. (\ref{z111}) and (\ref{z1}),
$\widehat{A}_0 = \widehat{A}_0^{(1)}$, where the latter is the multiplication operator by $|\cdot|$. Hence, the operator, cf. (\ref{22A}),
\begin{equation}
  \label{z705}
(A_{\alpha,0} k)(\eta) = - m |\eta|k(\eta) , \quad \mathcal{A}_{\alpha,0} = {\rm Dom} (A_{\alpha,0}) := \{k\in \mathcal{K}_\alpha: |\cdot| k \in \mathcal{K}_\alpha\}
\end{equation}
is the generator of the semigroup of $S_{\alpha, 0}(t)$, $t\geq 0$, defined by
\[
(S_{\alpha, 0}(t) k)(\eta) = \exp( - t m |\eta|) k(\eta).
\]
Clearly, for any $\varepsilon \in (0,1)$,
\begin{equation}
  \label{z750}
  \mathcal{A}_{\alpha,\varepsilon} \subset \mathcal{A}_{\alpha,0}.
\end{equation}
Let us show now that (\ref{z162}) holds. For $k\in \mathcal{K}_{\alpha'}$, by (\ref{z703}) and (\ref{z705}), we have
\begin{equation}
  \label{z706}
 (A_{\alpha, \varepsilon } - A_{\alpha, 0 })k = \varepsilon C k.
\end{equation}
For such $k$, we set
\begin{equation}
  \label{z707}
 u_{t} =  (S_{\alpha, \varepsilon}(t) - S_{\alpha, 0}(t))k.
\end{equation}
Then $u_0 =0$ and, cf. (\ref{z706}),
\begin{eqnarray}
  \label{z708}
\frac{d}{dt} u_{t} & = & S_{\alpha, \varepsilon}(t) A_{\alpha, \varepsilon} k - S_{\alpha, 0}(t) A_{\alpha, 0} k \\[.2cm]
& = & \varepsilon S_{\alpha, \varepsilon}(t) C k + A_{\alpha, 0}u_{t}. \nonumber
\end{eqnarray}
In the latter line, we have taken into account also (\ref{z750}). By (\ref{op-44}), one can define $C$ as a bounded linear operator
$C:\mathcal{K}_{\alpha'} \to \mathcal{K}_{\alpha''}$ for $\alpha'' \in (\alpha , \alpha')$. Then $Ck\in {\rm Dom} (A_{\alpha, \varepsilon})$, cf. (\ref{z34b}), and hence
\[
[0, + \infty) \ni t \mapsto \varphi_t := S_{\alpha, \varepsilon}(t) C k \in \mathcal{K}_\alpha
\]
is continuously differentiable in $\mathcal{K}_\alpha$ on $[0,+\infty)$. In view of (\ref{z1600}),
\[
\|S_{\alpha, \varepsilon}(t) \| \leq 1, \qquad {\rm for} \ \ {\rm all} \ \ t\geq 0 \ \ {\rm and} \ \ \varepsilon\in [0,1].
\]
Then
\begin{equation}
  \label{z34c}
\|\varphi_t\|_{\alpha} \leq \|C\|_{\alpha'\alpha} \|k\|_{\alpha'}.
\end{equation}
By Theorem 1.19, page 486 in \cite{Kato}, we have from the second line in (\ref{z708})
\begin{equation*}
  u_t = \varepsilon \int_0^t S_{\alpha, 0}(t-s) \varphi_s ds,
\end{equation*}
which by (\ref{z34c}) yields
\[
\sup_{t\in [0,T]} \|u_t\|_{\alpha} \leq \varepsilon T \|C\|_{\alpha'\alpha} \|k\|_{\alpha'}.
\]
Then (\ref{z162}) and (\ref{Mkappa}) follow from the latter estimate by (\ref{z707}) and (\ref{op-50}).
\end{myproof}

\section{The Kinetic Equation}
\label{KSec}

\subsection{Solving the equation}
\label{Solsec}

For the model which we consider, the kinetic equation is the following Cauchy problem in $L^\infty (\mathbb{R}^d)$, cf. (\ref{V-eqn-gen}),
\begin{equation}
  \label{K1}
\frac{d}{dt}\varrho_t = - m \varrho_t -  ( a^{-}\ast \varrho_t )\varrho_t + (a^{+}\ast  \varrho_t ), \qquad \varrho_t|_{t=0} = \varrho_0.
\end{equation}
Here, for an appropriate function $\varrho: \mathbb{R}^d \to \mathbb{R}$, we write
\begin{equation}
  \label{K2}
(a^{\pm}\ast \varrho ) (x) = \int_{\mathbb{R}^d} a^{\pm}(x-y) \varrho(y) dy = \int_{\mathbb{R}^d} \varrho(x-y)  a^{\pm}(y) dy,
\end{equation}
where $a^{\pm}$ are the kernels as in (\ref{Ra20}). The main peculiarity of (\ref{K1}) is that the solution of
(\ref{E5}) can be sought in the form
\begin{equation}
 \label{C239}
r_t (\eta) = e(\varrho_t, \eta):= \prod_{x\in \eta} \varrho_t(x),
\end{equation}
where $\varrho_t\in L^\infty (\mathbb{R}^d)$ is a solution of (\ref{K1}). Denote
\begin{eqnarray}
  \label{Va1}
 \varDelta^+ & = & \{  \varrho \in L^\infty (\mathbb{R}^d): \varrho (x) \geq 0 \ \ {\rm for} \ \ {\rm a. a.} \ \ x\}, \\[.2cm]
 \varDelta_b & = & \{ \varrho \in L^\infty (\mathbb{R}^d) :   \|\varrho \|_{L^\infty(\mathbb{R}^d)} \leq b\}, \qquad b>0,\nonumber \\[.2cm]
 \varDelta^{+}_b & = & \varDelta^{+} \cap \varDelta_b. \nonumber
\end{eqnarray}
\begin{lemma}
  \label{Vlm}
Let $\alpha^*$, $\alpha< \alpha^*$, and $T(\alpha)$ be as in Theorem \ref{2tm}. Set $b_0 = \exp(-\alpha^*)$ and $b = \exp(-\alpha)$.
Suppose that, for $t\in [0,T(\alpha))$, the problem in (\ref{K1}) with $\varrho_0 \in \varDelta^{+}_{b_0}$,
has a unique classical solution $\varrho_s\in \varDelta^{+}_{b}$ on $[0,t]$.
Then  the unique solution $r_s\in \mathcal{K}_{\alpha}$ of (\ref{E5}) with
$r_0 (\eta ) = e(\varrho_0 , \eta)$, see Theorem \ref{op-tm},  is given by (\ref{C239}).
\end{lemma}
\begin{proof}
First of all we note that, for a given $\alpha$, $e(\varrho, \cdot) \in \mathcal{K}_\alpha$ if and only if $\varrho \in \varDelta_b$ with $b= e^{-\alpha}$,
see (\ref{z300}) and (\ref{z30}). Now set $\tilde{r}_t = e(\varrho_t, \cdot)$ with $\varrho_t$ solving (\ref{K1}).
This $\tilde{r}_t$ solves (\ref{E5}), which can easily be checked by computing $d/dt$ and
employing the equation in (\ref{K1}). In view of the uniqueness as in Theorem \ref{op-tm},
we then have $\tilde{r}_s = r_s$ on $[0,t]$, from which it can be continued to $[0,T(\alpha))$ by repeating the same arguments on the interval $[t, 2t]$,
etc.
\end{proof}
As
(\ref{C239}) is the correlation function of the Poisson measure $\pi_{\varrho_t}$, the above lemma establishes the so called {\it chaos preservation} or {\it chaos propagation} in time as the most chaotic states are those corresponding to Poisson measures.
Now let us turn to solving (\ref{K1}).
\begin{theorem}
  \label{R1tm}
For arbitrary $\varrho_0 \in \varDelta^+$, the problem in (\ref{K1}) has a unique classical solution $\varrho_t \in \varDelta^{+}$ on $[0,+\infty)$.
\end{theorem}
\begin{proof}
For a certain $\epsilon \geq 0$, let us consider
\begin{equation}
  \label{300}
  u_t (x) = e^{-\epsilon t} \varrho_t (x), \qquad t\geq 0, \ \  x \in \mathbb{R}^d.
\end{equation}
Then $\varrho_t$ solves (\ref{K1}) if and only if $u_t$ solves the following problem
\begin{equation}
  \label{300u}
\frac{d}{dt} u_t = - (  m+\epsilon) u_t - e^{\epsilon t} (a^{-}\ast u_t) u_t  + (a^{+}\ast u_t),  \quad u_t|_{t=0} = \varrho_0.
\end{equation}
By integrating this differential problem we get the integral equation
\begin{eqnarray}
  \label{301u}
 u_t & = & \varrho_0 \exp\left(- (m+\epsilon) t - \int_0^t e^{\epsilon \tau}(a^{-}\ast u_\tau)d \tau \right) \\[.2cm]
 & + & \int_0^t (a^{+}\ast u_\tau) \exp\left(- (m+\epsilon) (t-\tau) - \int_\tau^t e^{\epsilon s}(a^{-}\ast u_s)d s \right) d \tau, \nonumber
\end{eqnarray}
which we will consider in the Banach space $\mathcal{C}_T$ of all continuous maps $u: [0,T]\to L^\infty(\mathbb{R}^d)$ with norm
\[
\|u\|_T := \sup_{t\in [0,T]} \|u_t\|_{L^\infty(\mathbb{R}^d)}.
\]
Here $T>0$ is a fixed parameter, which we choose later together with $\epsilon$. Then (\ref{301u}) can be written in the form
$ u = F(u)$, and hence the solution of (\ref{301u}) is a fixed point of the map defined by the right-hand side of this equation.
Set
$\mathcal{C}_T^+ = \{ u \in \mathcal{C}_T: \forall t\in [0,T]  \ u_t \in \varDelta^+\}$.
By (\ref{301u}) we have that, for each $u\in \mathcal{C}_T^+$ and for all $t\in [0,T]$, the following holds
\[
\|F(u)_t\|_{L^\infty(\mathbb{R}^d)} \leq\|\varrho_0 \|_{T} \exp[-
t(m +\epsilon)] + \|u\|_T \frac{\langle a^+\rangle}{m + \epsilon}
\bigg{(}1 - \exp[- t(m +\epsilon)] \bigg{)},
\]
where we consider $\varrho_0$ as a constant map from $[0,T]$ to
$L^\infty (\mathbb{R}^d)$. Now we set $\epsilon=0$ if $\langle
a^+\rangle \leq m$, and $\epsilon = \langle a^+\rangle - m$
otherwise. Then $\|F(u)_t\|_{L^\infty(\mathbb{R}^d)} \leq b$, and
hence $\|F(u)\|_{T} \leq b$ whenever $\max\{\|\varrho_0 \|_{T};
\|u\|_T \} \leq b$. Therefore, $F$ maps the positive part of a
centered at zero ball in $\mathcal{C}_T$  into itself. Let us now
show that $F$ is a contraction on such sets whenever  $T$ is small
enough. By means of the inequality
\[
\left\vert e^{-\alpha} -  e^{-\beta}  \right\vert \leq |\alpha -
\beta|, \qquad \alpha,\beta\geq 0,
\]
for fixed $b>0$ and $t\in [0,T]$, and
for $\varrho_0, u, \tilde{u} \in \mathcal{C}_T^+ $ such that
$\|{\varrho}_0\|_T, \|u\|_T, \|\tilde{u}\|_T \leq b$, we obtain from
(\ref{301u})
\begin{gather}
 \label{33}
\|F(u)_t - F(\tilde{u})_t \|_{L^\infty(\mathbb{R}^d)} \leq b
e^{-(m+\epsilon)t} \int_0^t e^{\epsilon \tau} \|U^{-}_\tau -
\tilde{U}^{-}_\tau\|_{L^\infty(\mathbb{R}^d)} d \tau \\[.2cm]
+ \int_0^t e^{-(m+\epsilon)(t-\tau)}\|U^{+}_\tau -
\tilde{U}^{+}_\tau\|_{L^\infty(\mathbb{R}^d)} d \tau \nonumber \\[.2cm] + \int_0^t
e^{-(m+\epsilon)(t-\tau)}\|\tilde{U}^{+}_\tau\|_{L^\infty(\mathbb{R}^d)}
\left(\int_\tau^t e^{\epsilon s} \|U^{-}_s -
\tilde{U}^{-}_s\|_{L^\infty(\mathbb{R}^d)} d s\right) d \tau,
\nonumber
\end{gather}
where $U^{\pm}_s := (a^{\pm}\ast u_s)$ , $\tilde{U}^{\pm}_s :=
(a^{\pm}\ast \tilde{u}_s)$, and hence, for all $s \in [0,T]$, we
have that
\[
\max\{\|U^{\pm}_s\|_{L^\infty(\mathbb{R}^d)};
\|\tilde{U}^{\pm}_s\|_{L^\infty(\mathbb{R}^d)} \} \leq b \langle
a^{\pm} \rangle.
\]
We use this in (\ref{33}) and obtain
\begin{gather*}
\|F(u)_t - F(\tilde{u})_t \|_{L^\infty(\mathbb{R}^d)} \leq b \left[
e^{-(m+\epsilon)t} + \frac{\langle a ^{+} \rangle}{m +
\epsilon}\left( 1 - e^{-(m+\epsilon)t}\right) \right] \\[.2cm] \times \int_0^t
e^{\epsilon \tau} \|U^{-}_\tau -
\tilde{U}^{-}_\tau\|_{L^\infty(\mathbb{R}^d)} d \tau + \frac{\langle
a ^{+} \rangle}{m + \epsilon}\left( 1 - e^{-(m+\epsilon)t}\right)
\|u - \tilde{u}\|_T.
\end{gather*}
The latter estimate yields
\begin{gather*}
\|F(u) - F(\tilde{u}) \|_T \leq q(T)\|u - \tilde{u}\|_T, \\[.2cm] q(T):=
b\langle a^{-} \rangle \int_0^T e^{\epsilon s} ds + \left( 1 -
e^{-(m+\epsilon)T}\right), \nonumber
\end{gather*}
where we have also taken into account that $\langle a ^{+} \rangle
\leq m + \epsilon$ due to our choice of $\epsilon$. Thus, for small
enough $T$, $F$ is a contraction, which yields that: (a)  the equations  in (\ref{300u}) and (\ref{301u}) are equivalent; (b) (\ref{301u}) has
a unique solution, $u\in \mathcal{C}_T^+$, such that $\|u\|_T \leq b$.  Now by
means of (\ref{300}) we return to the problem in (\ref{K1}) and
obtain that it has a positive solution $\varrho_t \in
L^\infty(\mathbb{R}^d)$, $t>0$ such that
\begin{equation}
 \label{34}
\|\varrho_s\|_{L^\infty(\mathbb{R}^d)} \leq \|\varrho_0\|_{L^\infty(\mathbb{R}^d)}\exp\left(- s(m- \langle a^{+} \rangle)\right), \quad s\in [0,t].
\end{equation}
Indeed, for $v_s:= e^{ms}\|\varrho_s\|_{L^\infty(\mathbb{R}^d)}$ from (\ref{300}) and (\ref{301u}) we get
\[
v_t \leq v_0 + \langle a^{+} \rangle \int_0^t v_s ds,
\]
which by the Gronwall inequality yields (\ref{34}).
The proof is completed.
\end{proof}
\subsection{Properties of the solution}
\label{Propsec}

In order to get additional tools for studying the solutions of  (\ref{K1}),
from now on we assume that the initial
conditions are taken from the set $C_{\rm b}
(\mathbb{R}^d)$ of bounded continuous functions $\phi: \mathbb{R}^d
\to \mathbb{R}$. Then the solution $\varrho_t$ will also belong to
$C_{\rm b} (\mathbb{R}^d)$ as this set is closed in
$L^\infty(\mathbb{R}^d)$ whereas the map in right-hand side of
(\ref{301u}) leaves it invariant. Thus, we can consider (\ref{K1})
in the Banach space obtained by equipping $C_{\rm b} (\mathbb{R}^d)$
with the supremum norm. Note that also the map $\phi \to
(a^{\pm}\ast \phi)$ leaves this spaces invariant -- by Lebesgue's
dominated convergence theorem this follows from the second equality
in (\ref{K2}). By $\tilde{\varDelta}$, $\tilde{\varDelta}^+$, and
$\tilde{\varDelta}^+_b$ we denote the intersections of the
corresponding sets defined in (\ref{Va1}) with $C_{\rm b}
(\mathbb{R}^d)$. Our main task is to understand which
properties of the model parameters, see (\ref{R20}) and
(\ref{Ra20}), imply that the solution in question is globally
bounded. If $m\geq \langle a^{+} \rangle$ then
$\|\varrho_t\|_{L^\infty(\mathbb{R}^d)} \leq
\|\varrho_0\|_{L^\infty(\mathbb{R}^d)}$ for all $t>0$, see
(\ref{34}). Thus, it is left to consider the case of $m < \langle
a^{+} \rangle$, in which $\varrho_t$ has exponential grows in $t$ if
$a^{-} \equiv 0$.

The following alternative situations are studied separately, cf. (\ref{z14}):
\begin{eqnarray}
  \label{separat}
& {\rm (i )} & \quad \exists  \theta >0  \quad a^{+}(x) \leq \theta a^{-}(x) \ \ {\rm for} \ {\rm a.a.}  \ x \in \mathbb{R}^d,\\[.2cm]
& {\rm (i i)} & \quad \Upsilon_\theta:= \left\{ x\in\mathbb{R}^d\,:\, a^{+}(x) > \theta a^{-}(x)\right\} \text{ of positive meas.}  \nonumber
\end{eqnarray}
\begin{theorem}
  \label{K2tm}
Let ${\rm (i)}$ in (\ref{separat})  hold. Then, for all $t>0$, the solution of (\ref{K1}) with $\varrho_0 \in C_{\rm b}(\mathbb{R}^d)$
lies in $\tilde{\varDelta}^+_{b}$ for some $b>0$. Furthermore, if
 $\varrho_0 \in \tilde{\varDelta}^+_{\theta-\delta}$, for some $\delta>0$, then  $\varrho_t \in \tilde{\varDelta}^+_\theta$ for all $t> 0$.
\end{theorem}
\begin{remark}
Theorem \ref{op-tm} establishes the existence of solution of (\ref{E5}) on a bounded time interval only. However, by Theorems \ref{K2tm} and \ref{350tm} below, as well as by Lemma \ref{Vlm}, the solution of (\ref{E5}) with $r_0= e(\varrho_0, \cdot)$ can be extended to the whole
$\mathbb{R}_{+}$ if the conditions of Theorems \ref{K2tm} or \ref{350tm} are satisfied.
\label{Decrk}
\end{remark}
\begin{myproof}{Theorem \ref{K2tm}}
Suppose that the second part of the statement holds true. Then if $\varrho_0$ is not in $\tilde{\varDelta}^+_{\theta-\delta}$,
we can increase $\theta$ until this condition is satisfied. For this bigger value of $\theta$, (i) in (\ref{separat}) clearly holds.
Now let us prove the second part.
Since the solution $\varrho_t (x)$ satisfies \eqref{34}, we have
$$
\|\varrho_s\|_{L^\infty(\mathbb{R}^d)} \leq (\theta-\delta)\exp\left(- s(m- \langle a^{+} \rangle)\right), \quad s\in [0,t].
$$
Hence, for small $t>0$, $\varrho_t(x) < \theta$ for all $x$. Then either the latter holds for all $t>0$, or there exists $t_0>0$ such that $\varrho_{t_0} (x_0) = \theta$ for some $x_0$, and $\varrho_{t} (x)$ is strictly increasing on $[t_0,t_0+\tau)$ for a small $\tau>0$. Then
\begin{equation}
  \label{35}
\left(\frac{d \varrho_{t_0}}{d t}\right)(x_0) = - m \theta - \left((\theta a^{-} - a^{+})\ast \varrho_{t_0} \right)(x_0) < 0.
\end{equation}
Thus, $\varrho_t (x_0)$ cannot increase at such a point.
\end{myproof}
Now let us turn to  case (ii) in (\ref{separat}). Define
\begin{equation*}
 f^{\pm} (\theta) = \int_{\Upsilon_\theta} a^{\pm}(x) dx, \qquad \theta>0.
\end{equation*}
By (ii) both functions are positive and non-increasing, and
\begin{equation*}
 g(\theta) := f^{+} (\theta) - \theta f^{-} (\theta) >0 ,
\end{equation*}
for all $\theta>0$. Hence, $f^{-} (\theta) < f^{+} (\theta)/\theta$.
\begin{theorem}
  \label{350tm}
Assume that there exists $\theta>0$ such that $g(\theta)< m$. Then the solution of (\ref{K1}) with
$\varrho_0 \in \tilde{\varDelta}^+_{\theta-\delta}$ for some $\delta>0$, lies in $\tilde{\varDelta}^+_{\theta}$ for all $t>0$.
\end{theorem}
\begin{proof}
Suppose that $\|\varrho_{t'}\|_{L^\infty(\mathbb{R}^d)}>\theta$, for some $t'>0$. Then, since $\|\varrho_t\|_{L^\infty(\mathbb{R}^d)}$ is continuous in $t$, one can choose small enough $\varepsilon>0$ with $\varepsilon<(m-g(\theta))/g(\theta)$ such that the set
$A_\varepsilon:=\{t>0 \, : \, \|\varrho_t\|_{L^\infty(\mathbb{R}^d)}=\theta+\varepsilon\}$ is nonempty. Since $\|\varrho_0\|_{L^\infty(\mathbb{R}^d)}<\theta$ and by the continuity arguments, one has $s:=\inf A_\varepsilon >0$ and $s\in A_\varepsilon$. Moreover, \begin{equation}\label{dop1}
\|\varrho_t\|_{L^\infty(\mathbb{R}^d)}<\theta+\varepsilon, \quad \text{for all}\ t\in(0,s)
\end{equation}
 (note that if $\|\varrho_{t_1}\|_{L^\infty(\mathbb{R}^d)}>\theta+\varepsilon$, for some $t_1\in(0,s)$, then there exists $t_2\in(0,t_1)\subset(0,s)$ with $\|\varrho_{t_2}\|_{L^\infty(\mathbb{R}^d)}=\theta+\varepsilon$
that contradicts the choice of $s$).

Since $\|\varrho_{s}\|_{L^\infty(\mathbb{R}^d)}=\theta+\varepsilon$, there exists $x\in\mathbb{R}^d$ such that $\varrho_s(x)\in(\theta,\theta+\varepsilon]$. For this $x$, $\varrho_0(x)<\theta$, therefore, the set $B_{s,x}:=\{t\in(0,s)\,:\, \varrho_t(x)=\theta\}$ is nonempty. By the continuity of $\varrho_t(x)$ in $t$, we have $\tau:=\sup B_{s,x}\in(0,s)$ and $\tau\in B_{s,x}$. Moreover, by a similar argument to that mentioned above,
\begin{equation}\label{dop2}
\varrho_t(x)>\theta, \quad \text{for all}\ t\in(\tau,s).
\end{equation}

Combining \eqref{dop1} and \eqref{dop2}, we get, for any $t\in(\tau,s)$ and for the chosen $x$,
\begin{eqnarray*}
\left(\frac{d \varrho_{t}}{d t}\right)(x) & < & - m  \theta  + \int_{\mathbb{R}^d} \left[ a^{+} (y) - \theta a^{-} (y)\right] \varrho_{t}(x-y) dy
\\[.2cm] & < &  - m  \theta  + (\theta+\varepsilon)\int_{\Upsilon_{\theta}} \left[ a^{+} (y) - \theta a^{-} (y)\right] dy\\[.2cm]
& = & \theta [- m + g(\theta)] +\varepsilon g(\theta)<0,
\end{eqnarray*}
by the choice of $\varepsilon$ maid above. Therefore, the function $\varrho_t(x)$ is decreasing in $t$ on $(\tau,s)$, hence, for all $t\in(\tau,s)$, $\varrho_t(x)<\varrho_\tau(x)=\theta$ that contradicts \eqref{dop2}. The contradiction shows that $\|\varrho_t\|_{L^\infty(\mathbb{R}^d)}\leq\theta$, for all $t>0$, that proves the statement.
\end{proof}
The condition crucial for the proof of Theorem \ref{350tm} is that $g(\theta) < m$ for $\theta$ such that $\varrho_0 \in \tilde{\varDelta}^+_{\theta-\delta}$. If $a^{-}$ has finite range, this holds under condition \eqref{globally} since
\[
\lim_{\theta \to +\infty} g(\theta)  = \int_{\Upsilon_{\infty}} a^{+}(x) dx, \quad \Upsilon_{\infty} := \bigcap_{\theta >0} \Upsilon_\theta   =
\{ x \in \mathbb{R}^d: a^{-} (x) = 0\}.
\]
Then, the solution $\varrho_t$ is globally bounded if
\begin{equation}
  \label{globally}
 \int_{\Upsilon_{\infty}} a^{+}(x) dx < m,
\end{equation}
which points to the role of the competition in the considered model -- if $a^{-} \equiv 0$, then the left-hand side of (\ref{globally}) is
just $\langle a^{+} \rangle$ and the condition in (\ref{globally}) turns into that of the sub-criticality in the contact model.\cite{KKP}
To illustrate this conclusion, let us consider the following example. For $r>0$, set $B_r = \{x\in \mathbb{R}^d: |x| \leq r\}$, and let
$\mathbb{I}_r$ and $|B_r|$ stand for the indicator and the Euclidean volume of $B_r$, respectively.
For
\begin{equation}
  \label{kernels}
a^{+} = \alpha \mathbb{I}_R, \quad a^{-} = \beta \mathbb{I}_r, \qquad R> r >0, \quad \alpha, \beta >0,
\end{equation}
we have
\[
g(\theta) = \left\{ \begin{array}{ll} \alpha |B_R| - \theta \beta |B_r|, \quad &{\rm if} \ \  \theta< \alpha /\beta;\\[.2cm]
 \alpha (|B_R| -  |B_r|), \quad &{\rm otherwise}
\end{array}\right.
\]
Hence, the condition of Theorem \ref{350tm} is satisfied if
\begin{equation}
\label{alpha}
\alpha (|B_R| -  |B_r|) < m.
\end{equation}
Case (ii) of (\ref{separat}) contains a subcase where one can get
more than the mere global boundedness established in Theorem
\ref{350tm}. From (\ref{301u}) it clearly follows that the solution
as in Theorem \ref{R1tm} is independent of $x$, i.e., is translation
invariant, if so is $\varrho_0$. This translation invariant solution
can be obtained explicitly. By setting $\varrho_t (x)\equiv \psi_t$ we obtain from (\ref{K1})  the following
\begin{equation}
\label{36}
\frac{d}{dt} \psi_t = (\langle a^{+} \rangle -m) \psi_t - \langle a^{-} \rangle \psi_t^2, \qquad \psi_t|_{t=0} = \psi_0,
\end{equation}
which is a Bernoulli equation. For $m > \langle a^{+} \rangle$, its solution decays to zero exponentially as $t \to +\infty$. For
$m = \langle a^{+} \rangle$, the solution is $\psi_t = \psi_0/(1 +  \langle a^{-} \rangle \psi_0 t)$, and hence also decays to zero as $t \to +\infty$.
For $m < \langle a^{+} \rangle$, we set
\begin{equation}
\label{37}
  q = \frac{ \langle a^{+} \rangle - m}{\langle a^{-} \rangle}.
\end{equation}
In this case the solution of (\ref{36}) has the form
\begin{equation}
  \label{38}
  \psi_t = \frac{\psi_0 q}{\psi_0 + (q-\psi_0) \exp(- q \langle a^{-} \rangle t)},
\end{equation}
which, in particular, means that $\psi_t \to q$ as $t \to +\infty$. Note that $\psi_t = q$ for all $t>0$ whenever $\psi_0 =q$.
\begin{theorem}
 \label{K4tm}
Suppose that $q>0$ and there exists $\varkappa^+>q$ such that $a^{+} (x) \geq \varkappa^+ a^{-} (x)$ for almost all $x\in \mathbb{R}^d$.
Let also the initial condition $\varrho_0 \in C_{\rm b}(\mathbb{R}^d)$ in (\ref{K1}) obey
\begin{equation}
  \label{40}
  0< \varkappa^{-} < \varrho_0 (x) < \varkappa^{+} < +\infty ,
\end{equation}
for some $\varkappa^{-}\in(0,q)$ and all $x\in \mathbb{R}^d$. Then, for each $x\in \mathbb{R}^d$ and $t>0$, the solution as in Theorem \ref{R1tm} obeys the bounds $\psi_t^{-} < \varrho_t(x) < \psi_t^{+}$, where $\psi_t^{\pm}$ are given in (\ref{38}) with $\psi_0=\varkappa^{\pm}$. Hence $\varrho_t (x) \to q$ in $C_{\rm b}(\mathbb{R}^d)$ as $t\to +\infty$.
\end{theorem}
The condition in Theorem \ref{K4tm} can be formulated as $\Upsilon_\theta = \mathbb{R}^d$ for all $\theta < q$. Its another form is
\begin{equation}
  \label{39}
 \frac{a^{+} (x)}{\langle a^{+} \rangle} \geq \left(1 - \frac{m}{\langle a^{+} \rangle} \right) \frac{a^{-} (x)}{\langle a^{-} \rangle},
\end{equation}
from which we see that the scale of the competition is irrelevant for the result stated in Theorem \ref{K4tm} to hold.
If $a^{+} (x)= \theta a^{-}(x)$, for some $\theta>0$ and almost all $x$, then (\ref{39}) holds for all $m\in [0, \langle a^{+} \rangle)$. If the competition has the range shorter than that of dispersal,  the mentioned homogenization
occurs at nonzero mortality $m$. For the example from (\ref{kernels}),  condition  (\ref{39}) holds if
\[
1 - \frac{m}{\langle a_{+} \rangle} \leq \left(\frac{r}{R} \right)^d,
\]
which is exactly the one given in (\ref{alpha}).
\begin{myproof}{Theorem \ref{K4tm}}
As $\psi_t^{-} < \varrho_t(x) < \psi_t^{+}$ clearly holds for $t=0$, by the continuity of the three functions of $t$ it holds for $t\in (0,\tau)$, for some $\tau>0$.
Write
\begin{gather*}
  \frac{d}{dt} \varrho_t (x) = - m \varrho_t (x)  + \int_{\mathbb{R}^d} a_t(x,y) \varrho_t (y) d y\\[.2cm]
  a_t (x, y) := a^{+} (x-y) - \varrho_t(x) a^{-} (x-y).
\end{gather*}
By the assumption $\varrho_t(x) < \psi^{+}_t < \varkappa^{+}$, $t \in (0,\tau)$, and hence $a_t$ is a positive kernel.
Also  $\varrho_t(x)> \psi^{-}_t$, $t\in (0,\tau)$, which yields
\begin{gather*}
   \frac{d}{dt} \varrho_t (x) \geq - m \varrho_t(x) - \langle a^{-} \rangle \psi^{-}_t \varrho_t (x) +   \langle a^{+} \rangle\psi^{-}_t.
\end{gather*}
Introduce $u_t (x) := \varrho_t (x) - \psi^{-}_t$ and obtain
\begin{gather*}
   \frac{d}{dt} u_t(x) \geq  - m u_t(x) - \langle a^{-} \rangle \psi^{-}_t u_t (x) \\[.2cm]
   = - (m + \langle a^{-} \rangle q ) u_t (x) + \langle a^{-} \rangle (q - \psi^{-}_t)u_t (x) \geq - \langle a^{+} \rangle  u_t (x),
\end{gather*}
where we have taken into account that $\psi^{-}_t < q$ for all $t>0$. The latter yields
\[
\varrho_t (x) - \psi^{-}_t \geq (\varrho_0 (x) - \varkappa) \exp( - t \langle a^{+} \rangle) , \qquad t \in (0, \tau).
\]
Hence, the estimate $\varrho_t (x) - \psi^{-}_t >0$ can be continued to arbitrary value of $t$ by repeating the above arguments with the same $\tau$. In a similar way, we obtain
\[
\psi^{+}_t - \varrho_t (x) \geq (\psi^{+}_0 - \varrho_0 (x))  \exp( - t [m + \langle a^{} \rangle\theta]),
\]
which completes the proof.
\end{myproof}

\subsection{Conclusion remarks}

\label{Conclsec}

The microscopic dynamics of the model considered here were first studied in \cite{Dima}, where weak$^*$ solutions of the problem in
(\ref{R4}) were shown to exist on $[0,+\infty)$ under the condition  which in our notations is  (\ref{z14}) plus also
$m > 16 \theta \langle a^{-} \rangle  +  4 \langle a^{+} \rangle$, see Theorems 4.6 and 5.1 on pages 309, 310. Afterwards, the results of
\cite{Dima} were used in \cite{FKK-MFAT} to derive and resume studying of the kinetic equation in (\ref{V-eqn-gen}).

An analog of (\ref{K1}) was non-rigorously deduced in \cite{Durrett} from a microscopic model on $\mathbb{Z}^d$.
Then this equation with $a^{+} = a^{-}$ was studied in \cite{Perthame}.
The case of equal kernels is covered by both Theorems \ref{K2tm} and \ref{K4tm}. According
to Theorem \ref{R1tm}, with no assumption on the parameters of the model we have the existence of the global evolution of $\varrho_t$,
which is in accord with Theorem \ref{op-tm}. A possible interpretation is that the mesoscopic description based on the
scaling applied here is insensitive to
the relationship between $a^{+}$ and $a^{-}$. This relationship is important
if one wants to get more detailed information, which is contained in Theorems \ref{K2tm} and \ref{K4tm}.

\paragraph{Acknowledgement}

This work was
financially supported by the DFG through SFB 701: ``Spektrale
Strukturen und Topologische Methoden in der Mathematik" and by the European Commission under the project
STREVCOMS PIRSES-2013-612669. The support from the ZiF
Research Group "Stochastic Dynamics: Mathematical Theory and
Applications" (Universit\"at Bielefeld) was also very helpful.



\begin{thebibliography}{00}





\bibitem{Bellomo} N. Bellomo, D. Knopoff, and J. Soler, On the difficult interplay between life, ``complexity", and mathematical sciences, {\it Math. Models Methods Appl. Sci.} {\bf 23} (2013) 1861--1913.

\bibitem{Berns} Ch. Berns, Yu. Kondratiev, Yu. Kozitsky, and O. Kutoviy, Kawasaki dynamics in continuum:
micro- and mesoscopic descriptions, {\it J. Dyn. Diff. Equat.} {\bf 25} (2013) 1027--1056.

\bibitem{Bogol} N. N. Bogoliubov, Problemy dinami\v{c}eskoi teorii v statisti\v{c}eskoi fizike (Russian)
[Problems of Dynamical Theory in Statistical Physics], Gosudarstv. Izdat. Tehn.--Teor.
Lit., Moscow-Leningrad, 1946.

\bibitem{BP1} B. M. Bolker and S. W. Pacala, {Using moment equations to understand stochastically driven spatial pattern formation in ecological systems,} {\it Theoret. Population Biol.} {\bf 52} (1997) 179--197.


\bibitem{BP3} B. M. Bolker, S. W. Pacala, and C. Neuhauser, Spatial dynamics in model plant communities: What do we really know? {\it The American Naturalist} {\bf 162} (2003) 135--148.


\bibitem{DL} U. Dieckmann and R. Law, {Relaxation projections and the method of moments,} {\it The Geometry of Ecological Interactions,} (Cambridge University Press, Cambridge, UK, 2000), pp. 412--455.

\bibitem{Dob} R. L. Dobrushin, Y. G. Sinai, and Y. M. Sukhov,
    {Dynamical systems of statistical mechanics,} {\it Itogi Nauki,} (VINITI, 1985), pp. 235--284; eng. transl. {\it Ergodic Theory with Applications to Dynamical Systems and Statistical Mechanics, II} ed. Ya. G. Sinai, (Encyclopaedia Math.~Sci., Springer, Berlin Heidelberg, 1989).

\bibitem{Durrett} R. Durrett, Crabgrass, measles, and gypsy moths: an introduction to modern probability, {\it Bull. Amer. Math. Soc.
(N.S.)} {\bf 18} (1988) 117--143

\bibitem{FKKo} D. Finkelshtein, Y. Kondratiev and Y. Kozitsky, Glauber dynamics in continuum: a constructive approach to evolution of states,
{\it Discrete Contin. Dyn. Syst.} {\bf 33} (2013) 1431--1450.

\bibitem{FKK-MFAT} D. Finkelshtein, Y. Kondratiev and O. Kutoviy, An operator approach to
Vlasov scaling for some models of spatial ecology, {\it  Methods Funct. Anal. Topology} {\bf 19} (2013) 108--126.





\bibitem{FKK} D. Finkelshtein, Y. Kondratiev and O. Kutoviy,
Semigroup approach to birth-and-death stochastic dynamics in continuum, {\it J. Funct. Anal.} {\bf 262} (2012) 1274--1308.

\bibitem{DimaN} D. L. Finkelshtein, Yu. G. Kondratiev, and O. Kutoviy, {Vlasov scaling for stochastic dynamics of continuous systems,} {\it J. Statist. Phys.} {\bf 141} (2010) 158--178.

\bibitem{Dima} D. L. Finkelshtein, Yu. G. Kondratiev, and O. Kutoviy, Individual based model with competition in spatial ecology, {\it SIAM J. Math. Anal.} {\bf 41} (2009) 297--317.


\bibitem{Dima2} D. L. Finkelshtein, Yu. G. Kondratiev, and M. J. Oliveira, {Markov evolution and hierarchical equations in the continuum. I: One-component systems,} {\it J. Evol. Equ.} {\bf 9} (2009) 197--233.


\bibitem{FM} N. Fournier and S. M{\'e}l{\'e}ard, {A microscopic probabilistic description of a locally regulated population and macroscopic approximations,} {\it Ann. Appl. Probab.} {\bf 14} (2004) 1880--1919.


\bibitem{Kato} T. Kato, {\it Perturbation Theory for Linear Operators.} (Second edition. Grundlehren der Mathematischen Wissenschaften, Band 132. Springer-Verlag, Berlin-New York, 1976).

\bibitem{Tobi} Yu. Kondratiev and T. Kuna, {Harmonic analysis on configuration space. I. General theory,} {\it Infin. Dimens. Anal. Quantum Probab. Relat. Top.} {\bf 5} (2002) 201--233.




\bibitem{KKP} Yu. Kondratiev, O. Kutoviy, and S. Pirogov, {Correlation functions and invariant measures in continuous contact model,} {\it Infin. Dimens. Anal. Quantum Probab. Relat. Top.} {\bf 11} (2008) 231--258.


\bibitem{Mu} D. J. Murrell, U. Dieckmann, and R. Law, {On moment closures for population dynamics in continuous space,} {\it J. Theoret. Biol.} {\bf 229} (2004) 421--432.

\bibitem{Neuhauser} C. Neuhauser, Mathematical challenges in spatial ecology,
{\it Notices of AMS} {\bf 48} (11) (2001) 1304--1314.




\bibitem{Obata} N. Obata, Configuration space and unitary representations of the group of diffeomorphisms, {\it RIMS K\^{o}ky\^{u}roku} {\bf 615} (1987) 129--153.

\bibitem{BCFKKO} O. Ovaskainen, D. Finkelshtein, O.Kutoviy, S, Cornell, B. Bolker, and   Yu. Kondratiev,    A general mathematical framework for the analysis of spatio-temporal point processes, {\em Theoretical Ecology} {\bf 7} (2014) 101--113.


\bibitem{Pazy} A. Pazy, {\it Semigroups of Linear Operators and Applications to Partial Differential Equations,}  Applied Mathematical Sciences, 44. (Springer-Verlag, New York, 1983).
\bibitem{Perthame} B. Perthame and P. E. Souganidis, Front propagation for a jump process model arising in spatial ecology, {\it Discrete Contin. Dyn. Syst.} {\bf 13} (2005)
1235--1246.

\bibitem{P}E. Presutti, {\it  Scaling Limits in Statistical Mechanics and Microstructures in  Continuum Mechanics,} (Theoretical and Mathematical Physics. Springer, Berlin. 2009.)




\bibitem{TV} H. R. Thieme and J. Voigt, {Stochastic semigroups: their construction by perturbation and approximation,} {\it Positivity IV---Theory and Applications,} eds. M. R. Weber and J. Voigt (Tech. Univ. Dresden, Dresden, 2006), pp. 135--146.

\bibitem{trev} F. Tr{\`e}ves, {\it Ovcyannikov Theorem and Hyperdifferential Operators,} Notas de Matem{\'a}tica, No.~46 (Instituto de Matem{'a}tica Pura e Aplicada, Conselho Nacional de Pesquisas, Rio de Janeiro, 1968).

\bibitem{zhao} Xiao-Qiang Zhao, {\it Dynamical Systems in Population Biology.} (CMS Books in Mathematics, Springer-Verlag, New York Inc., 2003).

\end{thebibliography}
\end{document}